\documentclass[a4paper,UKenglish,cleveref, autoref, thm-restate]{lipics-v2021}

\usepackage{mathtools,todonotes}

\newtheorem*{example*}{Example}

\newcommand{\nat}{\mathbb{N}}
\newcommand{\intg}{\mathbb{Z}}
\newcommand{\rel}{\mathbb{R}}
\newcommand{\rat}{\mathbb{Q}}

\newcounter{propcounter}
\newcounter{subpropcounter}[propcounter]
\renewcommand{\thepropcounter}{\arabic{propcounter}}
\renewcommand{\thesubpropcounter}{\thepropcounter\alph{subpropcounter}}

\newcounter{@propsave}
\newcounter{@subpropsave}

\newenvironment{proplist}
  {\setcounter{@propsave}{\value{propcounter}}
   \begin{list}{\thepropcounter.}
     {\usecounter{propcounter}
      \setcounter{propcounter}{\value{@propsave}}
      \leftmargin=1.5em}}
  {\end{list}}

\newenvironment{subproplist}
  {\setcounter{@subpropsave}{\value{subpropcounter}}
   \begin{list}{\thesubpropcounter.}
     {\usecounter{subpropcounter}
      \setcounter{subpropcounter}{\value{@subpropsave}}
      \leftmargin=3.0em}}
  {\end{list}}

\title{On existential B\"uchi arithmetic in two coprime bases} 

\author{Joris Nieuwveld}{University of Oxford, United Kingdom}{joris.nieuwveld@cs.ox.ac.uk
}{https://orcid.org/0009-0002-0339-1230}{}

\authorrunning{J. Nieuwveld} 

\Copyright{Joris Nieuwveld} 



\category{} 

\relatedversion{} 

\acknowledgements{This research was supported by the Glasstone Benefaction, University of Oxford [Violette and Samuel Glasstone Research Fellowships in Science 2025]}

\nolinenumbers 

\EventEditors{}
\EventNoEds{2}
\EventLongTitle{}
\EventShortTitle{}
\EventAcronym{}
\EventYear{}
\EventDate{}
\EventLocation{}
\EventLogo{}
\SeriesVolume{}
\ArticleNo{}
\begin{document}
\maketitle

\begin{abstract}
    For multiplicatively independent natural numbers $\alpha$ and $\beta$, Villemaire showed in 1992 that the first-order theory of Presburger arithmetic expanded with both B\"uchi predicates $V_\alpha$ and $V_\beta$ is undecidable, as it encodes multiplication.
    In recent years, Hieronymi and Schulz showed that Presburger arithmetic expanded with the weaker power predicates $\alpha^\nat = \{\alpha^n: n \in \nat\}$ and $\beta^\nat$ is also undecidable, while Karimov et al. showed that the existential fragment of this theory is decidable. 
    These results left open the natural problem of determining the decidability of the existential fragment of Villemaire's original expansion.
    We settle this question for coprime $\alpha$ and $\beta$.
    Specifically, we give a quantifier-elimination argument that proves the decidability of the existential fragment of $\mathsf{FO}(\intg;<,+, V_\alpha, V_\beta)$.
\end{abstract}

\section{Introduction}
For nearly a century, \emph{Presburger arithmetic}, the first-order theory of the integers with order and addition $\mathsf{FO}(\intg;<,+)$, has been studied, and countless modern works use its decidability. 
In 1929, Presburger established decidability using a quantifier-elimination method, showing that Hilbert's dream of effectively establishing the truth of any given arithmetic statement is attainable in the restricted setting of linear arithmetic. 
However, unable to express multiplication and many other arithmetic constructions, Presburger arithmetic is rather weak and far from Hilbert's unachievable dream in its full sense.
Nowadays, we understand this infeasibility well: Church showed a few years later that Presburger arithmetic expanded with multiplication is undecidable.\footnote{See Haase~\cite{haase2018survival} for a more detailed overview of the history of Presburger arithmetic.}

The introduction of automata in the 1950s provided a new perspective.
A natural number written in a fixed base $b$ can be seen as a finite word or string, and one can formulate addition and order as string operations.
That is, for natural numbers $x,y,z$, the expressions $x < y$ and $x + y = z$ hold if and only if certain automata accept $(x,y)$ and $(x,y,z)$, respectively, in their base-$b$ expansions.
This revelation led to another decision procedure for Presburger arithmetic: Turn a relation $R \subseteq \nat^k$ definable in Presburger arithmetic into a finite-state automaton that encodes the language of $R$ and determine whether this language is empty.\footnote{In this introduction we state results using the universes $\intg$ and $\nat$ as most are more natural in one universe: $\intg$ from an arithmetic perspective and $\nat$ from an automata perspective.}

Subsequently, B\"uchi~\cite{buchi1960weak} introduced a new predicate $V_b(x,y)$, the \emph{B\"uchi predicate}, and its expansion of Presburger arithmetic $\mathsf{FO}(\intg;<,+,V_b)$ called \emph{B\"uchi arithmetic}.
B\"uchi's predicate $V_b(x,y)$ holds if and only if $y$ is the largest power of $b$ dividing $x$. 
For example, $V_3(12, 3)$ holds while $V_2(12, 2)$ does not (as $4$ is a larger power of $2$ dividing $12$).
B\"uchi's predicate $V_b(x,y)$ is recognised by a finite-state automaton that uses the base-$b$ expansions of $x$ and $y$ as input.
Automata-theoretic methods can be used to show that B\"uchi arithmetic is decidable.
In a sense, B\"uchi arithmetic captures the full power of automata: The B\"uchi–Bruy\`ere theorem (see~\cite{bruyere1994logic}) states that a subset $S \subseteq \nat^r$ is definable in base-$b$ B\"uchi arithmetic if and only if $S$ is $b$-recognisable, i.e., some finite-state automaton with alphabet $\{0,\dots,b-1\}^r$ has language~$S$.

The encoding of $V_b$ above relies on natural numbers being represented in a single base $b$.
An automaton can translate to a second base $\tilde{b}$ when $b$ and $\tilde{b}$ are multiplicatively dependent,\footnote{That is, there are no integers $k, \ell$ such that $b^k\tilde{b}^\ell = 1$.} while multiplicatively independent $\tilde{b}$ and $b$ do not allow such a translation. 
In particular, in 1992, Villemaire~\cite{villemaire1992theory} showed the following result.
\begin{theorem}
    Let $\alpha$ and $\beta$ be multiplicatively independent.
    Then $\mathsf{FO}(\intg;<,+,V_\alpha,V_\beta)$ defines multiplication, and is therefore undecidable.
\end{theorem}
The set of powers of $b$, denoted $b^\nat$, is $b$-recognisable.
Presburger arithmetic expanded with a \emph{power predicate} $b^\nat$ is strictly less expressive than B\"uchi arithmetic, i.e., certain $b$-recognisable sets cannot be defined in $\mathsf{FO}(\intg;<,+,b^\nat)$.
A few years after Villemaire's result, B\`es~\cite{Bes1997} showed that replacing one of the two B\"uchi predicates with the weaker power predicate does not change the result: $\mathsf{FO}(\intg;<,+, V_\alpha, \beta^\nat)$ defines multiplication and is undecidable.

It took until 2021 for Hieronymi and Schulz~\cite{hieronymi2022strong} to establish that $\mathsf{FO}(\intg;<,+,\alpha^\nat,\beta^\nat)$ is undecidable when $\alpha$ and $\beta$ are multiplicatively independent.
However, they encoded the Halting Problem for Turing machines directly, as Schulz~\cite{schulz2025undefinability} also showed that multiplication cannot be defined in this structure.
Recently, Karimov et al.~\cite{karimov2025decidability} showed that the decidability of restricted fragments of such theories is not completely hopeless. 
Using quantifier elimination, they showed that the existential fragment of $\mathsf{FO}(\intg;<,+,\alpha^\nat,\beta^\nat)$ is decidable when $\alpha$ and $\beta$ are multiplicatively independent.\footnote{Although they included $0$ and $1$ in their language, we omit them as they are existentially definable. Zero is the unique number satisfying $x+x=x$ and $1$ is the unique number in both $\alpha^\nat$ and $\beta^\nat$.}

The current paper extends the result of~\cite{karimov2025decidability} from power predicates to the more expressive B\"uchi predicates under the (stronger) condition that $\alpha$ and $\beta$ are coprime.

\begin{theorem}\label{thm:main theorem}
    Let $\alpha, \beta \ge 2$ be coprime. 
    Then the existential fragment of $\mathsf{FO}(\intg;<,+,V_\alpha,V_\beta)$ is decidable.
\end{theorem}

As $x$ is a power of $z$ if and only if $V_z(x,x)$ holds, the theorem above entails the result of Karimov et al. for coprime $\alpha$ and $\beta$. 
In Section~\ref{Sec:conclusion}, we discuss the difficulties for general multiplicatively independent bases.
To emphasise this distinction, lemmas that do not rely on coprimality explicitly use multiplicative independence.

The rest of the paper is devoted to proving Theorem~\ref{thm:main theorem}. 
In Section~\ref{sec:reduction to system}, we reduce the decision problem to a system of inequalities and modular equations involving powers $z^n$ of $\alpha$ and $\beta$ and multiples $xz^n$ of these powers.
We reduce these kinds of general systems to those satisfying a list of properties, which we give and establish in Section~\ref{sec:outline}.
The necessary constructions are gently introduced in Section~\ref{sec:outline} and formally defined in Section~\ref{sec:levels and bottom family}, where we also show they are computable.
We finalise the quantifier-argument in Section~\ref{sec:Finalising the argument}.
For readability and flow, we included many examples, while moving most proofs to appendices.

\section{Preliminaries}\label{sec:prelims}
Let $\mathbf{0}$ denote the (column) vector $(0,\dots,0)$ whose dimension should be clear from context, and similarly, $\mathbf{1}$ denotes $(1,\dots,1)$.
We use bold letters for vectors and their entries: $\mathbf{x} = (\mathbf{x}_1,\dots,\mathbf{x}_t)$.
Matrices and vectors are rational except when stated otherwise. For a matrix $A$ and $k \le \ell$, $A_{k :\ell}$ is the restriction of $A$ to its $k$th to $\ell$th column.

\subsection{Logic and Decidability}\label{sec:logic}
A structure $\mathbb{M}$ consists of a \emph{universe} $U$, constants $c_1,\dots,c_k$, functions $f_1\colon U^{\mu_1} \to U,\dots,\allowbreak f_\ell\colon U^{\mu_\ell} \to U$, and predicates $P_1\subseteq U^{\nu_1},\dots,P_m\subseteq U^{\nu_m}$ (with $\mu_i,\nu_i \in \nat$).
The \emph{language} of $\mathbb{M}$ consists of all well-formed first-order formulas built from constants, functions, predicates, and the symbols $=$, $\lnot$, $\land$, $\lor$, $\exists$, and $\forall$. 
A \emph{theory} of $\mathbb{M}$ is a subset of its language.
The \emph{first-order theory} of $\mathbb{M}$, denoted $\mathsf{FO}(U;c_1,\dots,c_k,f_1,\dots,f_\ell,P_1,\dots,P_m)$, is the set of all such \emph{sentences} (formulas without free variables) in the language of $\mathbb{M}$ that hold in $\mathbb{M}$.
The \emph{existential fragment} $\exists\mathsf{FO}(U;c_1,\dots,c_k,f_1,\dots,f_\ell, P_1,\dots, P_m)$ of $\mathbb{M}$ consists of all existential sentences in its first-order theory, i.e., all sentences $\exists x_1,\dots,x_k \colon \varphi(x_1,\dots,x_k)$ where $\varphi$ is a quantifier-free formula in the language of~$\mathbb{M}$.
A theory is \emph{decidable} if there is an algorithm that decides whether a given sentence is in this theory, and \emph{undecidable} otherwise.
We note that the first-order theory of \emph{Presburger arithmetic}, $\mathsf{FO}(\intg;<,+)$, is decidable.

\subsection{Linear programming}\label{sec:lin programming prelims}
We use~\cite{schrijver1998theory} for many of these definitions.
Let $A$ be a matrix and $\mathbf{v}$ a vector of appropriate dimensions with rational coefficients and $\sim \in \{>,\ge\}$.
We study the rational or integer-valued, $t$-dimensional vectors $\mathbf{x}$ that satisfy $A\mathbf{x} \sim \mathbf{v}$. We refer both to the abstract space defined by the inequality and the set of vectors $\mathbf{x}$ in the appropriate ring as a \emph{polyhedron} of \emph{ambient dimension} $t$.
This should be clear from context.
Let $T := \{t_1,\dots,t_r\} \subseteq \{1,\dots,t\}$.
Then $B(\mathbf{x}_{t_1},\dots,\mathbf{x}_{t_r}) \sim \mathbf{w}$ is the \emph{projection} of the polyhedron $A\mathbf{x} \sim \mathbf{v}$ onto the coordinates from $T$ if the following holds for all $\mathbf{x} \in \rat^r$: $B\mathbf{x} \sim \mathbf{w}$ if and only if there is a $\tilde{\mathbf{x}} \in \rat^t$ satisfying $A\tilde{\mathbf{x}}\sim\mathbf{v}$ and $\mathbf{x}_1 = \tilde{\mathbf{x}}_{t_1},\dots,\mathbf{x}_r = \tilde{\mathbf{x}}_{t_r}$. 

\begin{lemma}\label{lem:lin programming trick}
    With the notation above, one can effectively compute a matrix $B$ and vector $\mathbf{w}$ such that $B\mathbf{x}' > \mathbf{w}$ is the projection of $A\mathbf{x} > \mathbf{v}$ onto the coordinates $T$.
\end{lemma}

Let $A\mathbf{x} \ge \mathbf{0}$ denote the \emph{recession cone} of $A\mathbf{x} \sim \mathbf{v}$.
Its \emph{intrinsic dimension} is the maximum number of $\rat$-independent vectors in its recession cone (which is thus independent of $\mathbf{v}$).
The Minkowski-Weyl theorem gives a finite set of vectors $\mathbf{a}_1,\dots,\mathbf{a}_r$ such that every $\mathbf{x}$ in the recession cone is a $\rat_{\ge 0}$-linear combination of $\mathbf{a}_1,\dots,\mathbf{a}_r$.
One can compute these vectors $\mathbf{a}_i$ using the double description method.
The \emph{orthogonal complement} of a recession cone $A\mathbf{x} \ge \mathbf{0}$ is the $\rat$-vector space $U = \{\mathbf{u} : \mathbf{u}\cdot\mathbf{x} = 0\text{ whenever } A\mathbf{x} \ge \mathbf{0}\}$.
Clearly, $\dim(U)$ and the intrinsic dimension sum to the ambient dimension.
If $\mathbf{a}_i$ are the vectors obtained from the Minkowski-Weyl theorem, $\mathbf{y} \in U$ if and only if $\mathbf{y}$ is orthogonal to every $\mathbf{a}_i$.
\begin{lemma}\label{lem:compute basis orth space}
    Given a matrix $A \in\intg^{s \times t}$, one can compute a basis for the linear subspace $U = \{\mathbf{u} \in \rat^t : \mathbf{u} \cdot \mathbf{x} = 0\text{ whenever } A\mathbf{x} \ge \mathbf{0}\}$.
\end{lemma}
\begin{definition}\label{def : thin}
    The polyhedron $A\mathbf{x} \sim \mathbf{v}$ is thin if there is a number $f$ and vector $\mathbf{u} \in \rat^t \setminus \{\mathbf{0}\}$ such that $|\mathbf{u}\cdot\mathbf{x}| < f$ whenever $A\mathbf{x} \sim \mathbf{v}$.
\end{definition}
The following observation is proven in Appendix~\ref{app:Baker}.
\begin{lemma}\label{lem:thinness computable}
    One can determine whether a polyhedron is thin, and if so, one can compute a witness $\mathbf{u}$ and $f$. 
\end{lemma}

\subsection{Kronecker's theorem on Diophantine approximation and Baker's theorem}
Let $\{x\}$ denote the \emph{fractional part} of $x$, i.e., $\{x\} \in [0, 1)$ is the real number such that $\{x\} \equiv x \pmod{1}$.
The following is a variation on Kronecker's theorem on Diophantine approximation, of which a simple version states that $\{n\theta: n \in \nat\}$ is dense in $(0,1)$ when $\theta$ is irrational (see, e.g.,~\cite[Theorem 439]{hardy1979introduction}).
The following consequence is proven in Appendix~\ref{app:Baker}.
\begin{theorem}\label{thm:kronecker dense}
    Let $\alpha, \beta \in \intg_{\ge2}$ be multiplicatively independent and $M, N \in \nat$. Then the set $\{\alpha^n/\beta^m : m \in \intg_{>M}, n \in \intg_{>N}\}$ is dense in $\rel_{>0}$.
\end{theorem}
Baker's theorem is a powerful tool from number theory. To define it, we need to define the \emph{logarithmic (Weil) height} $h(x)$ of a rational number $x$, which, when writing $x = p/q$ in lowest terms, is just $\max(\log|p|,\log|q|)$. 
We give the following special case of Baker's theorem, which follows from the work of Matveev~\cite{matveev2000explicit}.
\begin{theorem}\label{thm:baker}
    Let $\alpha,\beta \ge 2$ be multiplicatively independent integers.
    Then there is a computable constant $C$ with the following property. If $k_\alpha,k_\beta \in \nat$, $c \in \rat_{\ne0}$, $N = \max(3, |k_\alpha|, |k_\beta|)$, and $\Lambda = c \alpha^{k_\alpha}\beta^{k_\beta}-1$, then $\Lambda = 0$ or $|\Lambda| > N^{-Ch(c)}\,.$
\end{theorem}

In~\cite{karimov2025decidability}, Baker's theorem was applied to certain \emph{exponential Diophantine equations}.
\begin{theorem}\label{thm:eq with only pure powers}
    Let $\alpha, \beta \ge 2$ be multiplicatively independent, $r \ge 0$, $c_1,\dots,c_r,c \in \intg$, and $z_1,\dots,z_r \in \{\alpha, \beta\}$. 
    Then one can compute $\intg$-modules $W_1,\dots,W_\ell$ of $\intg^r$ such that
    \begin{equation}\label{eq:Soda 2025}
        c_1z_1^{k_1} + \cdots + c_rz_r^{k_r} = c
    \end{equation}
    if and only if $(k_1,\dots,k_r) 
    \in \nat^r$ is in some $W_j$.
    For each $\intg$-module $W_j$, one can compute a partition $\mathcal{I}_1,\mathcal{I}_2,\mathcal{I}_3$ of $\{1,\dots,r\}$, a map $\sigma \colon \mathcal{I}_2 \to \mathcal{I}_1$, and $f_i \in \nat$ for all $i \in \mathcal{I}_2\cup\mathcal{I}_3$ such that $(k_1,\dots,k_r) \in W_j$ if and only if $k_i = k_{\sigma(i)} + f_i$ for all $i\in \mathcal{I}_2$ and $k_i = f_i$ for all $i\in \mathcal{I}_3$.
\end{theorem}
We add the following addendum to this theorem whose proof is in Appendix~\ref{app:Baker}.
\begin{lemma}\label{lem::soda observation}
    If at least one $c_i$ is non-zero, we must have $|\mathcal{I}_1| < r$.
\end{lemma}
We need two technical generalisations of Theorem~\ref{thm:eq with only pure powers}.
In these, the constants $c_i$ are variable but still relatively small.
The proofs are in Appendix~\ref{app:Baker}.
\begin{lemma}\label{lem:Baker helper reduce}
    Let $\alpha, \beta \ge 2$ be coprime, $C_1 \in \rat_{>1}$, $C_2, C_3 \in \rat$, $Z\in\{\alpha,\beta\}$, and $p_\alpha,p_\beta \in \rat[X]$. 
    We can compute a $p \in \rat[X]$ such that for all $x_\alpha,x_\beta \in \intg$ and $m, k_\alpha, k_\beta, \ell_\alpha, \ell_\beta \in \nat$ we have that $\ell_Z - m < p(\log\ell_Z)$ whenever
    \begin{enumerate}
        \item $\alpha^{\ell_\alpha}/\beta^{\ell_\beta} \in (1/C_1, C_1)$;
        \item $|x_z| < C_2 z^{\ell_z-k_z}$ for $z \in \{\alpha, \beta\}$;
        \item $\ell_z - k_z < p_z(\log\ell_z)$ for $z \in \{\alpha, \beta\}$; and
        \item $0< \big|x_\alpha \alpha^{k_\alpha} + x_\beta \beta^{k_\beta} \big| < C_3Z^{m}$.
    \end{enumerate}
\end{lemma}

\begin{lemma}\label{lem:Baker to eliminate bottom family modular constraints}
    Let $\alpha, \beta \ge 2$ be coprime, $z_1,\dots,z_r\in\{\alpha,\beta\}$, $c_1,\dots,c_r \in \intg_{\ne0}$, $c \in\intg$, $C_1\in \rat_{>1}$, $C_2 \in \rat$, and $p_\alpha,p_\beta\in \rat[X]$. 
    We can compute $N \in \nat$ such that for all $k_\alpha,k_\beta, \ell_\alpha, \ell_\beta$, $n_1,\dots,n_r \in \nat$ and $x_\alpha,x_\beta\in\intg$ such that $x_\alpha \alpha^{k_\alpha} + x_\beta \beta^{k_\beta}+c_1z_1^{n_1} + \cdots + c_rz_r^{n_r} + c = 0$,
    \begin{enumerate}
        \item\label{Baker cond 1} $z_1^{n_1} < \cdots < z_r^{n_r} < \alpha^{k_\alpha}, \beta^{k_\beta}$ and $z^{k_z} \le z^{\ell_z}$ for $z \in \{\alpha, \beta\}$;
        \item\label{Baker cond 2} $\alpha^{\ell_\alpha}/\beta^{\ell_\beta} \in (1/C_1, C_1)$;
        \item\label{Baker cond 3} $|x_z| < C_2 z^{\ell_z - k_z}$  for $z \in \{\alpha, \beta\}$;
        \item\label{Baker cond 4} $\ell_z - k_z < p_z(\log\ell_z)$ for $z \in \{\alpha, \beta\}$; and
        \item\label{Baker cond 5} No non-empty proper subsum of the multiset $\{x_\alpha \alpha^{k_\alpha}, x_\beta \beta^{k_\beta},c_1z_1^{n_1}, \ldots, c_rz_r^{n_r}\}$ is zero,
    \end{enumerate}
    we have that $k_\alpha,k_\beta, \ell_\alpha, \ell_\beta,n_1,\dots,n_r < N$.
\end{lemma}

\section{Proof of the main theorem}
\subsection{Reduction to systems of (in)equalities}\label{sec:reduction to system}
To start our proof of the main result, we transform an existential formula $\varphi$ in $\mathsf{FO}(\intg;<,+, V_\alpha, V_\beta)$ to a system of equalities, inequalities, and modular constraints with two kinds of unknowns: exponential ones of the form $z_j^{n_j}$ and linear-exponential ones of the form $x_{i,j}z_j^{n_j}$. 
Recall that a function $f \colon \intg^s \to \intg$ is called \emph{affine} if there are integers $a_1,\dots,a_s,a$ such that $f(x_1,\dots,x_s) = a+\sum_{i=1}^s a_i x_i$.
The B\"uchi predicate $V_z$ can be written as
\begin{equation}\label{eq:to solve}
        \exists n \in \nat, x'\in \intg \colon \bigvee_{1 \le a < z} x = (zx'+a)z^n \land y  =z^n\,.
\end{equation}

\begin{proposition}\label{prop:reduction to a system}
    Let $\alpha,\beta \in \intg_{\ge 2}$. Then  $\exists \mathsf{FO}(\intg;<,+,V_{\alpha},V_{\beta})$
    is decidable if one can determine whether for given ${\sim_i}\in \{>, =\}$, $c_i, c_i',d_i' \in \intg$, $r, s,s', s_1,\dots,s_r \in \nat$, affine functions $f_{i,j},f_{i,j}' \colon \intg^{s_j}\to\intg$, and $z_j \in \{\alpha, \beta\}$ it holds that
    \begin{equation}\label{eq:sys before mod}
    \begin{split}
        \exists n_1,\dots,n_r \in \nat, \mathbf{y}_1 &\in \intg^{s_1},\dots,\mathbf{y}_r \in \intg^{s_r}\colon\bigwedge_{i=1}^s \:\:\sum_{j=1}^r f_{i,j}(\mathbf{y}_j)z_j^{n_j} \,\sim_i\, c_i\\
        &\land \bigwedge_{i=1}^{s'} \:\:\sum_{j=1}^r f_{i,j}'(\mathbf{y}_j)z_j^{n_j} \equiv\, c_i'\pmod{d_i'}\,.
    \end{split}
    \end{equation}
\end{proposition}
\begin{proof}
    We rewrite an existential formula $\varphi$ in this structure into conjunctive normal form. 
    Each atom is of the form $x = 0$, $x > 0$, or $V_z(x, y)$, where $z \in \{\alpha,\beta\}$ and $x$ and $y$ are linear combinations of existentially quantified variables in $\varphi$.
    We reduce to the existential positive fragment by writing the negations of these atoms as a positive combination:
    \begin{itemize}
        \item $\lnot V_z(x, y) \Longleftrightarrow  x = 0 \lor \exists y' \colon V_z(x, y') \land \lnot y = y'$
        \item $\neg  x = 0 \Longleftrightarrow x > 0 \lor -x > 0$;
        \item $\neg x > 0 \Longleftrightarrow -x+1 > 0$.
    \end{itemize}
    Substitute atoms $V_z(x, y)$ with the equivalent formula~\eqref{eq:to solve} to obtain the following system:
    \begin{align*}
        &\exists n_1,\dots,n_r \in \nat, \mathbf{y}_1 \in \intg^{s_1},\dots,\mathbf{y}_r \in \intg^{s_r},y_1,\dots,y_t\in\intg\colon \\
        &\quad\quad \bigwedge_{i=1}^s \:\:\sum_{j=1}^r f_{i,j}(\mathbf{y}_j)z_j^{n_j} + \sum_{j=1}^t a_{i,j}y_j \,\,\sim_i\,\, c_i\,,
    \end{align*}
    where $r, \sim_i, c_i,s,s_1,\dots,s_r, f_{i,j},z_j$ are as in the hypothesis, $t \in \nat$, and $a_{i,j} \in \intg$.
    
    It remains to eliminate the variables $y_i$, which we do using standard quantifier techniques. 
    Assume we want to eliminate $y_t$. 
    At certain steps, we introduce modular conditions
    \begin{equation*}
        \sum_{j=1}^r f_{i,j}(\mathbf{y}_j)z_j^{n_j} + \sum_{j=1}^t a_jy_j \equiv c_i \pmod{d}\,.
    \end{equation*}
    By splitting into the cases where $(e_0,\dots,e_t) \in \{0,\dots,d-1\}^{t+1}$ and $e_0 \equiv \sum_{j=1}^r f_{i,j}(\mathbf{y}_j)z_j^{n_j} \pmod{d}$ and $e_j \equiv y_j \pmod{d}$ for $1 \le j \le t$ such that $e_0 + \sum_{j=1}^ta_je_j \equiv c \pmod{d}$ and making the substitutions $y_j \mapsto d y_j + e_j$ we eliminate $y_j$ from these modular constraints.
    
    If there is some $i$ such that $\sim_i$ is an equality and $a_{i,t}$ is non-zero, rename such that $i = s$. 
    Use this equality to eliminate $y_t$ from the equations and inequalities, multiplying each with an appropriate constant to make all numbers integer-valued. 
    Turn this equality into the modular condition  
    \begin{equation*}
        \sum_{j=1}^r f_{s,j}(\mathbf{y}_j)z_j^{n_j} + \sum_{j=1}^{t-1} a_{s,j}y_j \,\,\equiv \,\, c_s \pmod{|a_{s,t}|}
    \end{equation*}
    and add it to the conjunction of modular constraints. 
    Otherwise, $a_{i,t} = 0$ whenever $\sim_i$ is an equality. 
    Compute $\mathcal{I}_-$ and $\mathcal{I}_+$, the subsets of $\{1,\dots,s\}$ such that $a_{i,t} < 0$ and $a_{i,t} > 0$, respectively.
    If $\mathcal{I}_- = \emptyset$, take $y_t$ arbitrarily large and reduce to the (in)equalities labelled $\{1,\dots,s\} \setminus \mathcal{I}_+$ in which $a_{i,t}$ is always zero. 
    We likewise assume that $\mathcal{I}_+ \ne \emptyset$.
    By multiplying these inequalities with appropriate constants, we can assume that $a\coloneqq |a_{i,t}|$ is the same for all $i\in\mathcal{I}_-\cup\mathcal{I}_+$.
    Let $x_i$ be a shorthand for $\sum_{j=1}^r f_{i,j}(\mathbf{y}_j)z_j^{n_j} + \sum_{j=1}^{t-1} a_{i,j}y_j - c_i$ such that we have the inequalities $x_i - ay_t > 0$ for $i \in \mathcal{I}_-$ and $x_i + ay_t > 0$ for $i \in \mathcal{I}_+$ as the only (in)equalities where $y_t$ plays a role.
    By renaming and splitting into cases, we can assume that $1 \in \mathcal{I}_-$, $2 \in \mathcal{I}_+$, $x_1 < x_i$ for $i \in \mathcal{I}_-$, and $x_2 < x_i$ for $i \in \mathcal{I}_+$.
    Then $y_t$ should satisfy $x_1 > ay_t > -x_2$, which can be shown by showing that either $x_1+ x_2 \ge a+2$ or if $0 < x_1 + x_2$ is small and $x_1$ and $x_2$ are in particular 
    congruence classes modulo $a$.
    
    Hence, we have eliminated $y_t$.
    The proposition follows.
\end{proof}

In the last proof, we observed that statements $\exists y \in \intg \colon x + a y = 0$ with $a \ne 0$ can be turned into modular equations $x \equiv 0 \pmod{|a|}$. 
The same principle applies when $y = x_{i,j}z_j^{n_j}$: For $a \ne 0$ and $x \in \intg$, we have that $\exists x_{i,j} \in \intg \colon x + ax_{i,j}z_j^{n_j} = 0$ is equivalent to $x \equiv 0 \pmod{|a|z_j^{n_j}}$.
If an equation necessarily does not have an unknown $x_{i,j}z_j^{n_j}$ with a non-zero coefficient, e.g., $\alpha^{n_1} + 7\alpha^{n_2} - 3\beta^{n_3} = 6$, Theorem~\ref{thm:eq with only pure powers} lets us reduce to a system with strictly fewer variables.
The following proposition formalises this process.

\begin{proposition}\label{prop:inequalities + mods}
    Let $\alpha,\beta \ge 2$. Then $\exists\mathsf{FO}(\intg;<,+,V_\alpha,V_\beta)$ is decidable when for all $r,s,s',s'',t_1, \dots,t_r\in\nat$, $t = \sum_{i=1}^r (1+t_i)$, bijections $F \colon \{1,\dots,t\} \to \{(i,j) : 1 \le i \le r, 0 \le j \le t_i\}$,
    $z_1,\dots,z_r \in \{\alpha,\beta\}$, $A \in \intg^{s \times t}$, $\mathbf{v} \in \intg^s$, $c_{i,k}, c_{i',k}', c_i, c_{i'}' \in \intg$, $d_i, d_{i'}' \in \intg_{\ge 0}$, $k_i \in \{1,\dots,r\}$ for $1 \le i \le s'$, $1 \le i'\le s''$ and $1 \le k \le t$, one can determine whether there are $n_1,\dots,n_r \in \nat$ and $x_{i,j} \in \intg_{\ge1}$, $\mathbf{x} \in \intg_{\ge 1}^t$ such that 
    \begin{equation*}
        \mathbf{x}_k = 
        \begin{cases}
            z_j^{n_j} & \text{if $F(k) = (0,j)$}\\
            x_{i,j}z_j^{n_j} &\text{if $F(k) = (i,j)$ and $i \ge 1$}
        \end{cases}
    \end{equation*}
     and
    \begin{equation}\label{eq:sys with mods}
        A\mathbf{x}
        > \mathbf{v}
        \: \land \: \bigwedge_{i=1}^{s'}  \sum_{k=1}^tc_{i,k}\mathbf{x}_k\equiv c_i \pmod{d_i z_{k_i}^{n_{k_i}}} \:\land\:\, \bigwedge_{i=1}^{s''}  \sum_{k=1}^tc_{i,k}'\mathbf{x}_k\equiv c_i' \pmod{d_i'}\,.
    \end{equation}
    Moreover, we can assume that $\mathbf{x}_1 > \cdots > \mathbf{x}_t > 0$ whenever $A\mathbf{x} > \mathbf{v}$ and that $c_{i,k'} = 0$ for $1 \le i \le s'$ and $1 \le k < k'< t$ such that $F(k) = (0,k_i)$ and $F(k') = (i',j')$ for some $i'\ge 1$.
\end{proposition}
\begin{proof}
    We reduce from Proposition~\ref{prop:reduction to a system} by translating~\eqref{eq:sys before mod} into~\eqref{eq:sys with mods}.
    Unpacking the affine functions, for each $1 \le j \le t$, we have $t_j+1$ variables: $z_j^{n_j}, x_{1,j}z_j^{n_j},\dots,x_{t_j,j}z_j^{n_j}$.
    As the bijection $F$ suggests, we introduce variables $\mathbf{x}_1,\dots,\mathbf{x}_t$ corresponding to the unknowns in~\eqref{eq:sys before mod}, a matrix $A$, a vector $\mathbf{v}$, and integers $c_{i,j}', c_i'$, and $s \in \nat$ and split the inequalities and equalities.
    By splitting into cases, we can assume that $x_{i,j} = 0$ (giving strictly fewer variables) or $x_{i,j}$ (by splitting in the positive and negative cases and changing signs of the coefficients accordingly). Again, by splitting into cases and possibly renaming, we can assume that $\mathbf{x}_1 > \cdots > \mathbf{x}_t > 0$, giving the second-to-last claim.
    
    We show that for a system of equalities, inequalities, and modular constraints (involving both constant and variable moduli), we can eliminate an equality.
    Fix an equality $\sum_{k=1}^t c_{i,k}\mathbf{x}_k = c_i$.
    There are two cases. 

    First, assume that $c_{i,k} = 0$ whenever $F(k) = (i',j')$ with $i'\ge1$. 
    That is, when $c_{i,k}$ is non-zero, $\mathbf{x}_k$ is of the form $z^n$ with $z \in\{\alpha,\beta\}$, giving an exponential Diophantine equation on which we can apply Theorem~\ref{thm:eq with only pure powers}. 
    Therefore, we can compute a finite number of $\intg$-modules $W_1,\dots, W_w \subseteq \intg^r$ such that this equation holds if and only if $(n_1,\dots,n_r) \in\nat^r$ is in a $W_k$.
    Split into cases.
    Compute the map $\sigma$, partition $\mathcal{I}_1, \mathcal{I}_2, \mathcal{I}_3$ of $\{1,\dots,r\}$, and $f_j \in \intg$ for all $j \in \mathcal{I}_2 \cup \mathcal{I}_3$ such that $(n_1,\dots,n_r)$ is in $W_k$ if and only if for all $j \in \mathcal{I}_2$ we have that $n_j = n_{\sigma(j)} + f_j$ and $z_j = z_{\sigma(j)}$ for all $j\in\mathcal{I}_2$ and $n_j = f_j$ for all $j\in\mathcal{I}_3$.
    Making these suggested substitutions gives a system with at most the same number of variables and one fewer equation.

    Otherwise, there is a $1 \le k \le t$ such that $F(k) = (i',j)$ with $i' \ge 1$ and $c_{i,k} \ne 0$, and without loss of generality, we can assume $k$ to be maximal. 
    Using Gaussian elimination, eliminate $\mathbf{x}_k$ from all inequalities in $A\mathbf{x} > \mathbf{v}$, the equations $\{1,\dots,s'\}\setminus\{i\}$, and all already formed modular relations. 
    Then using that
    \begin{equation}\label{eq:eliminate variable for mod}
        \exists x_{i,j} \in \intg\colon \sum_{k'=1}^t c_{i,k'}\mathbf{x}_{k'}= c_i \Longleftrightarrow \sum_{\substack{k'=1\\k'\ne k}}^t c_{i,k'}\mathbf{x}_{k'}\equiv c_i \pmod{|c_{i,k}| z_j^{n_j}}\,,
    \end{equation}
    we have eliminated $\mathbf{x}_k$. 
    By construction, the last claim also holds.
\end{proof}

For the remainder of the paper, we study the polyhedron $A\mathbf{x} > \mathbf{v}$, where (the boldface) $\mathbf{x}_k$ denotes the integer corresponding to a \emph{power variable}.
We call a variable $\mathbf{x}_k$ of the form $z_j^{n_j}$ a \emph{pure power variable} and a variable of the form $x_{i, j}z_j^{n_j}$ a \emph{linear power variable}.
Thus, (the non-boldface) $x_{i,j}$ denotes the \emph{linear part} of a linear power variable $x_{i,j}z_j^{n_j}$.
In both cases, we call $z_j$ the (\emph{exponential}) \emph{base} of $\mathbf{x}_k$ and $z_j^{n_j}$ the \emph{exponential part}.
We have two kinds of modular constraints: those with a \emph{variable modulus}, i.e., those whose modulus is of the form $d_iz_{k_i}^{n_{k_i}}$, and those with a \emph{constant modulus}, i.e., those whose modulus is of the form~$d_i'$. 

\begin{remark}\label{rem:reduce when equality}
    When we encounter a system as in Proposition~\ref{prop:inequalities + mods} and an equation $\sum_{k=1}^tc_k\mathbf{x}_k = c$, where not all of $c_1,\dots,c_t,c$ are zero (call such an equation a \emph{non-trivial equation}), we can reduce to a system with strictly fewer variables by applying the proof of Proposition~\ref{prop:inequalities + mods} \emph{mutatis mutandis}.
    Similarly, if we encounter a non-power variable, i.e., a variable which is not a power variable, we can instantly eliminate it as in Proposition~\ref{prop:reduction to a system}.
\end{remark}

\subsection{Properties for our system}\label{sec:outline}
This section outlines the first step of our algorithm for solving the system of inequalities and modular equations in powers of $\alpha$ and $\beta$ established in Proposition~\ref{prop:inequalities + mods}.
Our algorithm reduces these systems to finitely many systems that satisfy certain properties from which we can eliminate many variables at once.
In this process, $t$ (the number of variables) will never increase, but the number of modular constraints might.
For this reason, we converted the equations into modular constraints earlier.
We introduce a few notions before listing these properties and the remainder of our reduction.

We treat the power variables $\mathbf{x}_k$ both as numbers that satisfy a system of linear inequalities implying that $\mathbf{x}_1 > \dots > \mathbf{x}_t > 0$ and modular constraints, and we treat these $\mathbf{x}_k$ as numbers which have this specific form of being a pure or linear power variable. 

\textbf{Levels.}
We start by defining \emph{levels} from a linear algebra perspective as non-empty sets of power variables of roughly equal size.
We flesh this further out in Section~\ref{sec:levels and bottom family} and treat this here mostly by example.
Using the strict ordering  $\mathbf{x}_1 > \dots > \mathbf{x}_t > 0$, we say that $\mathbf{x}_k$ and $\mathbf{x}_{k+1}$ are on the same level when $\mathbf{x}_k/\mathbf{x}_{k+1} < C$ has to hold for some constant $C$.
Else, $\mathbf{x}_k$ and $\mathbf{x}_{k+1}$ are not on the same level.
The first level contains the largest variables, the second the next largest variables, etc.
Let us go through an example.
\begin{example}\label{ex:outline}
    Let $t = 4$ and take the following system of inequalities:
    \begin{equation}\label{ex:three variables}
        \begin{split}
            \mathbf{x}_1 > \mathbf{x}_2 > \mathbf{x}_3 > \mathbf{x}_4 &> 0\\
            \mathbf{x}_1 - 2\mathbf{x}_3 + 3\mathbf{x}_4&> 0 \\
            -\mathbf{x}_1 + 2\mathbf{x}_3 - \mathbf{x}_4 &> 0\,.
        \end{split}
    \end{equation}
    The last two inequalities are equivalent to $|\mathbf{x}_1 - 2\mathbf{x}_3 + 2\mathbf{x}_4| < \mathbf{x}_4$.
    The first level includes $\mathbf{x}_1$, the largest power variable.
    The last two inequalities imply that $\mathbf{x}_1/\mathbf{x}_3$ is bounded, and so $\mathbf{x}_1$, $\mathbf{x}_2$, and $\mathbf{x}_3$ are on the first level. 
    As for any $c > 3$, $\mathbf{x}_1 = (2c-2)\mathbf{x}_4$, $\mathbf{x}_2 = (2c-3)\mathbf{x}_4$, $\mathbf{x}_3 = c\mathbf{x}_4$ satisfies the system of inequalities and $\mathbf{x}_3/\mathbf{x}_4$ gets arbitrarily large as $c$ increases, $\mathbf{x}_1$, $\mathbf{x}_2$, and $\mathbf{x}_3$ form the first level while $\mathbf{x}_4$ alone forms the second.
\end{example}

Equivalently, the first level consists of $\mathbf{x}_1,\dots,\mathbf{x}_\lambda$ when $\lambda \ge 1$ is the smallest number such that the recession cone $A_{1:\lambda}(\mathbf{x}_1, \dots, \mathbf{x}_\lambda) \ge \mathbf{0}$ contains a vector $(\mathbf{x}_1,\dots,\mathbf{x}_\lambda) > 0$. 
Here we call $B_1 := A_{1:\lambda}$ the first \emph{induced subpolyhedron}.
One continues this process inductively.
Project $\mathbf{x}_1,\dots,\mathbf{x}_\lambda$ out and apply the same method on the remaining cone, getting the next level, say the $\ell$th.
Then we also obtain the $\ell$th induced subpolyhedron $B_\ell$.
On each level $\ell$, there is a maximal linear subspace $U_\ell$ orthogonal to $B_\ell$, the \emph{orthogonal complement of the $\ell$th level}. 
Thus, for any $\mathbf{u} \in U_\ell$, $B_\ell\mathbf{u} = \mathbf{0}$. 
Let $\mathbf{u}_{\ell,1},\dots,\mathbf{u}_{\ell,p_\ell}$ denote a basis of $U_\ell$.

\begingroup
\renewcommand{\theexample}{} 
\makeatletter
\def\@thmcountersep{} 
\makeatother
\begin{example}[Continuation of Example~\ref{ex:outline}]
    Turn the inequalities~\eqref{ex:three variables} into a matrix form:
    \begin{equation*}
        \begin{pmatrix} 1 & -1 & 0& 0\\0 & 1 & -1& 0 \\0 &0 & 1 & -1\\0 & 0 & 0& 1\\1 & 0 & -2& 3 \\ -1 & 0 & 2& -1\end{pmatrix}
        \begin{pmatrix}     \mathbf{x}_1\\\mathbf{x}_2\\\mathbf{x}_3\\\mathbf{x}_4\end{pmatrix} \ge \mathbf{0}\,.
    \end{equation*}
    Due to its bottom two rows, the restriction of this polyhedron to the first or first two coordinates has no strictly positive element. 
    However, taking $(\mathbf{x}_1,\mathbf{x}_2,\mathbf{x}_3) = (4,3,2)$, we have that $A_{1:3}(\mathbf{x}_1,\mathbf{x}_2,\mathbf{x}_3) =(1, 1, 2, 0, 0, 0) \ge \mathbf{0}$. 
    Thus, again, $\mathbf{x}_1$, $\mathbf{x}_2$, and $\mathbf{x}_3$ form the first level. 
    For all $(\mathbf{x}_1,\mathbf{x}_2,\mathbf{x}_3)$ satisfying $A_{1:3}(\mathbf{x}_1,\mathbf{x}_2,\mathbf{x}_3) \ge \mathbf{0}$, we have that $\mathbf{x}_1 - 2\mathbf{x}_3 = 0$.
    Hence, $(1, 0, -2)$ is in $U_1$, and we can show that this vector actually generates~$U_1$.
\end{example}
\endgroup

Computing the partition of the variables into levels, the induced polyhedra, their orthogonal complements, and bases is straightforward linear algebra.

\textbf{The bottom family.} 
We introduce the \emph{bottom family} next.
Levels indicate which variables can be arbitrarily far apart, but we also need a more subtle distinction. 
Certain restrictions still keep variables close compared to their actual size.
For example, it could be $\mathbf{x}_1/\mathbf{x}_2 < \log(\mathbf{x}_1)^{C}$ for some constant $C$.
We return to our example.

\begingroup
\renewcommand{\theexample}{} 
\makeatletter
\def\@thmcountersep{} 
\makeatother
\begin{example}[Continuation of Example~\ref{ex:outline}]
    Assume that $\mathbf{x}_1 = 2^{n_1}$, $\mathbf{x}_2 = x_{1,2}2^{n_2}$, $\mathbf{x}_3 = 3^{n_3}$, and $\mathbf{x}_4 = 2^{n_2}$.
    The bottom two inequalities imply that $|2^{n_1} - 2\cdot 3^{n_3}| < 3\cdot 2^{n_2}$. Using Baker's theorem on $\frac{1}{2^{n_1}}|2^{n_1} - 2\cdot 3^{n_3}| = |2\cdot 3^{n_3}/2^{n_1} - 1| $ gives that 
    \begin{equation*}
        \mathbf{x}_4/\mathbf{x}_1 = 2^{n_2 - n_1} > \frac{1}{3}|2\cdot 3^{n_3}/2^{n_1} - 1| > \frac{1}{3}\max(3,n_1,n_3)^{-c}
    \end{equation*}
    for some computable constant $c$ whenever $2^{n_1} \ne 2\cdot 3^{n_3}$.
    Hence, as $3^{n_3}/2^{n_1} \approx 1/2$, we have that $\mathbf{x}_4/\mathbf{x}_1 > \frac{1}{3}n_1^{-c}$ for all computable large $n_1$.
    Thus, $\mathbf{x}_1/\mathbf{x}_4 < \log_2(\mathbf{x}_1)^c$, which is polynomial in $\log(\mathbf{x}_1)$. 

    As a second quirk, we note that as $\mathbf{x}_2 = x_{1,2}2^{n_2} < \mathbf{x}_1$, we have that $x_{1,2}$ has to be small. 
    To be exact, $x_{1,2} < 2^{n_1-n_2}$ is polynomial in $n_1$.
\end{example}
\endgroup

In Section~\ref{sec:levels and bottom family}, we define the bottom family using linear algebra.
As a consequence of the proof of the upcoming property~\ref{step 3c}, one can compute a polynomial $p$ such that for all $\mathbf{x}_k$ in the bottom family, $\mathbf{x}_1/\mathbf{x}_k < p(\log\mathbf{x}_1)$  whenever $A\mathbf{x} > \mathbf{v}$. 
A power variable has \emph{small coefficients} if its exponential part is in the bottom family. 
In our example, the bottom family consists of $\mathbf{x}_1,\dots,\mathbf{x}_4$ and these variables have small coefficients.
These coefficients are either constant (for pure power variables) or can only grow polynomially in~$n_1$.

\textbf{Properties we need to fulfil.} Let us list all properties that our system, levels, bottom family, and the bases of the orthogonal spaces $U_\ell$ have to satisfy. 
The proofs of these lemmas are in Appendix~\ref{app:properties}.

\begin{proplist}
\item\label{step 1} If $\mathbf{x}_k = x_{i,j}z_j^{n_j}$ is a linear power variable and $\mathbf{x}_{k'} = z_j^{n_j}$ is its accompanying pure power variable, then $\mathbf{x}_{k'} < \mathbf{x}_k$. Moreover, $t \ge 1$.
\end{proplist}

This first property sounds technical but is quite natural: Else, $x_{i,j}z_j^{n_j} < z_j^{n_j}$ with $z_j^{n_j} > 0$ and $x_{i,j} > 0$ being an integer. That is of course absurd.
Secondly, when $t = 0$, the formula is variable-free, and testing satisfiability is trivial.

\begin{proplist}
\item\label{step 2} The polyhedron $A\mathbf{x} > \mathbf{v}$ is not thin.
\end{proplist}

If the polyhedron was thin, then, in the notation of Section~\ref{sec:lin programming prelims}, there are a computable non-zero vector $\mathbf{u}$ and a number $f$ such that $|\mathbf{u}\cdot \mathbf{x}|< f$ whenever $A\mathbf{x} > \mathbf{v}$.
Hence, one of the non-trivial equations $\mathbf{u}\cdot \mathbf{x} = -f+1,\dots,\mathbf{u}\cdot \mathbf{x}=f-1$ has to hold, reducing the number of polyhedra to deal with.
Combined with the observation for property~\ref{step 1}, we obtain the following lemma.

\begin{lemma}\label{lem:properties 1,2}
    When one of properties~\labelcref{step 1,step 2} does not hold, we can either reduce to a system with strictly fewer variables or directly determine whether the given system has a solution.
\end{lemma}

\begin{proplist}
\item Modular constraints with variable modulus $\sum_{k=1}^tc_{i,k}\mathbf{x}_k\equiv c_i \pmod{d_i z_{k_i}^{n_{k_i}}}$ satisfy the following.
\end{proplist}

\begin{subproplist}
\item\label{step 3a} When $c_{i,k} \ne 0$ and $\mathbf{x}_k$ is a linear power variable, $\mathbf{x}_k$ has exponential part at least~$z_{k_i}^{n_{k_i}}$.
\end{subproplist}

Property~\ref{step 3a} is satisfied at the start by Proposition~\ref{prop:inequalities + mods}.
If we achieve a system as in the proposition by some other reduction, then the proof of the proposition shows we can rewrite the system and assume this property holds.
Thus, we obtain the following lemma (whose proof is not in an appendix as it is immediate).
\begin{lemma}\label{lem:property 3a}
    We can reduce to systems where property~\ref{step 3a} holds.
\end{lemma}

\begin{subproplist}
\item\label{step 3b} Each $d_i$ is a power of $z_{k_i}$.
\end{subproplist}

Property~\ref{step 3b} does not give a direct reduction but can be satisfied cheaply using the Chinese Remainder theorem.
For example, if $z_{k_i} = 6$, $d_i = 10$ and $x$ is some expression in power variables.
Then $x \equiv 7 \pmod{10\cdot 6^{n_{k_i}}}$ is equivalent to $x \equiv 7 \pmod{5} \land x \equiv 7 \pmod{2\cdot 6^{n_{k_i}}}$, where $2$ divides $6$. Hence, multiplying with $3 = 6/2$ gives that the second modular constraints is equivalent to $3x \equiv 21 \pmod{6\cdot 6^{n_{k_i}}}$.
The following lemma formalises this method.
\begin{lemma}\label{lem:d_i power of z_(k_i) 3b}
    We can reduce to a system where property~\ref{step 3b} holds, and this reduction preserves property~\ref{step 3a} and the polyhedron $A\mathbf{x} > \mathbf{v}$.
\end{lemma}

\begin{subproplist}
\item\label{step 3c} The power variable $z_{k_i}^{n_{k_i}}$ is not in the bottom family.
\end{subproplist}

This property is the most work to achieve and relies on iterative applications of Baker's theorem (Lemma~\ref{lem:Baker helper reduce}) and its consequence Lemma~\ref{lem:Baker to eliminate bottom family modular constraints}.
We prove that either the modular constraint is trivially satisfied or that it effectively bounds $n_{k_i}$, leading to a direct reduction in the number of variables.

\begin{lemma}\label{lem:step 3c}
    Assume properties~\ref{step 1},~\ref{step 2},~\ref{step 3a},~\ref{step 3d}, and~\labelcref{step 4a,step 4b,step 4c}.
    If property~\ref{step 3c} does not hold, we can reduce to finitely many systems with strictly fewer variables, or all modular constraints violating property~\ref{step 3c} trivially hold.
\end{lemma}

\begin{subproplist}
\item\label{step 3d} Assume $\mathbf{x}_k$ has exponential part $\mathbf{x}_{k'} = z_{k_i}^{n_j}$ and that $\mathbf{x}_{k''} = z_{k_i}^{n_{k_i}}$.
Then $c_{i,k} = 0$ whenever $\mathbf{x}_{k'} > \mathbf{x}_{k''}$, or $\mathbf{x}_{k'} = \mathbf{x}_{k''}$ and $\mathbf{x}_k$ is a linear power variable.
\end{subproplist}

Satisfying property~\ref{step 3d} is more technical and is done simultaneously with the upcoming property~\ref{step 4c}.
We postpone sketching its use and proof until after this list.

\begin{proplist}
\item The levels are organised as follows.
\end{proplist}

\begin{subproplist}
\item\label{step 4a} Pure power variables with the same exponential basis are not on the same level.
\item\label{step 4b} A pure power variable and linear power variable with the same exponential part are not on the same level.
\end{subproplist}

The two properties above are very similar as are the arguments that they can be satisfied.
For the first, if $z^{n_1}$ and $z^{n_2}$ are on the same level, $z^{n_1-n_2} \in (1/C, C)$, where $C$ is the computable constant from the construction of the levels.
Then there are only finitely many possibilities for $n_1 - n_2$, allowing us to make a substitution and simplify matters.
Similarly, when $\mathbf{x}_1 = xz^n$ and $\mathbf{x}_2 = z^n$, then $x < C$, and again we can reduce the number of variables by splitting into the finitely many options for $x$.

\begin{lemma}\label{lem:step 4a, 4b}
When one of properties~\labelcref{step 4a,step 4b} does not hold, we can reduce to finitely many systems with strictly fewer variables.
\end{lemma}

\begin{subproplist}
\item\label{step 4c} If $z \in \{\alpha,\beta\}$ and $\mathbf{u} \ne \mathbf{0}$ is in the orthogonal space $U_\ell$, then the coefficient in~$\mathbf{u}$ of some power variable with exponential base $z$ is not zero.
\end{subproplist}

It remains to explain and establish properties~\ref{step 3d} and~\ref{step 4c}. 
We obtain these properties level by level while preserving property~\ref{step 3a}.
Let us go through an example.

\begingroup
\renewcommand{\theexample}{} 
\makeatletter
\def\@thmcountersep{} 
\makeatother
\begin{example}[Continuation of Example~\ref{ex:outline}]
    With the system~\eqref{ex:three variables}, assume that
    $\mathbf{x}_1 = 3^{n_1}$, $\mathbf{x}_2 = 2^{n_2}$, $\mathbf{x}_3 = x3^{n_3}$, and $\mathbf{x}_4 = 3^{n_3}$.
    Recall that $\mathbf{x}_1,\mathbf{x}_2,\mathbf{x}_3$ and $\mathbf{x}_4$ form the two levels of this system. 
    Moreover, the orthogonal subspace of the first level $U_1$ is generated by $(1, 0, -2)$.

    Let us start with property~\ref{step 4c}. 
    It is violated on the first level as $(1, 0, -2) \in U_1$, implying that $|3^{n_1} - 2x3^{n_3}|$ has to be small, where each term here has the same exponential base. 
    As $3^{n_3} < 3^{n_1}$, $3^{n_1} - 2x3^{n_3}$ is always a multiple of $3^{n_3}$, which puts quite some restraints on this difference. 
    We solve this problem by introducing a new variable $\mathbf{x}_3' = y3^{n_3} = \mathbf{x}_1 - 2\mathbf{x}_3$. 
    This introduction gives a new equality with which we eliminate $\mathbf{x}_3 = \frac{1}{2}\mathbf{x}_1-\frac{1}{2}\mathbf{x}_3'$.
    To remember that $\mathbf{x}_1 - \mathbf{x}_3' \equiv 0 \pmod {2\mathbf{x}_4}$, we add this modular constraint with variable modulus to our system. 
    Next, we have to ensure that $y > 0$, which we do by splitting into cases according to its sign $s$ and mapping $\mathbf{x}_3$ to $s\mathbf{x}_3$.
    Then, depending on $s$, we obtain the following system:
    \begin{equation*}
        \begin{pmatrix} 1 & -1 & 0& 0\\-1/2 & 1 & s/2& 0 \\1/2 & 0& -s/2 & -1\\0 & 0 & 0& 1\\0 & 0 & s& 3 \\ 0 & 0 & -s& -1\end{pmatrix}
        \begin{pmatrix}     \mathbf{x}_1\\\mathbf{x}_2\\\mathbf{x}_3'\\\mathbf{x}_4\end{pmatrix} \ge \mathbf{0}\,.
    \end{equation*}
    The bottom two rows indicate that $s = -1$ is the only viable option. 
    We are done, save for having to insert $\mathbf{x}'_3$ back into the ordering $\mathbf{x}_1 > \mathbf{x}_2 > \mathbf{x}_4 > 0$. 
    Equality would add a non-trivial equality, allowing a direct reduction by Remark~\ref{rem:reduce when equality}, so we only have to consider the four possibilities of strict inequalities. 
    If $\mathbf{x}_2 > \mathbf{x}_3'$, $\mathbf{x}_1$ and $\mathbf{x}_2$ form the first level (take $\mathbf{x}_1 = 3$ and $\mathbf{x}_2 = 2$), and this first level satisfies property~\ref{step 4c}.
    Else, as the bottom two rows imply that $\mathbf{x}'_3$ and $\mathbf{x}_4$ are on the same level, all four power variables are on the same level, which violates properties~\labelcref{step 4a,step 4b}, allowing us to reduce the number of variables.
    In general, property~\ref{step 4c} follows from repeatedly applying this procedure.

    Next, for property~\ref{step 3d}, take the modular constraint $\mathbf{x}_1 + \mathbf{x}_3 \equiv 0 \pmod{9\mathbf{x}_4}$.
    That is, $3^{n_1} + x3^{n_3} \equiv 0 \pmod{9\cdot 3^{n_3}}$.
    Then this property is violated by both $\mathbf{x}_1$ and $\mathbf{x}_3$ as $n_1 > n_3$. 
    
    First, let us satisfy it for $\mathbf{x}_1$.
    If $3^{n_1} > 3\mathbf{x}_4$, then $3^{n_1} \equiv 0 \pmod{9\mathbf{x}_4}$, and so we consider the options $\mathbf{x}_1 = 3\mathbf{x}_4$ (leading to fewer variables) and $\mathbf{x}_1 > 3\mathbf{x}_4$ (which we add to our inequalities and leads to the equivalent modular constraint $x3^{n_3} \equiv 0 \pmod{9\cdot 3^{n_3}}$).
    Thus, we dealt with $\mathbf{x}_1$.

    For $\mathbf{x}_3$, it is a little bit more complicated.
    Consider the options $\mathbf{x}_3 = 9\mathbf{x}_3'+i\mathbf{x}_4$ for $i = 0,\dots,8$ to get that
    $i\mathbf{x}_4 \equiv 9\mathbf{x}_3' + i\mathbf{x}_4 \equiv \mathbf{x}_3 \equiv 0 \pmod{9\mathbf{x}_4}$, allowing us to remove it from the modular constraint to obtain $i\mathbf{x}_4 \equiv 0 \pmod{9\mathbf{x}_4}$, which now satisfies the property.
    Then we make the substitution for $\mathbf{x}_3$ to $\mathbf{x}_3'$ and have to insert $\mathbf{x}_3'$ in the ordering $\mathbf{x}_1>\mathbf{x}_2>\mathbf{x}_4>0$.

    The difficulty in the following lemma is showing this entire procedure terminates.
\end{example}
\endgroup
\begin{lemma}\label{lem:replacements terminate 4c}
    Assume properties~\ref{step 1},~\ref{step 2},~\ref{step 3a},~\ref{step 3b},~\ref{step 4a}, and~\ref{step 4b} hold. Then we can reduce to a finite number of systems that satisfy properties~\ref{step 3a},~\ref{step 3d}, and~\ref{step 4c}.
\end{lemma}

\subsection{Levels and the bottom family}\label{sec:levels and bottom family}
In this section, we formally define levels, bases for their orthogonal subspaces, the bottom family, and power variables with small coefficients.

Recall that we have $t$ variables $\mathbf{x}_1 > \cdots > \mathbf{x}_t > 0$ where each $\mathbf{x}_k$ is (uniquely) of the form $z_j^{n_j}$ (a pure power variable) or $x_{i,j}z_j^{n_j}$ (a linear power variable).
We will actually use a looser assumption: Each $\mathbf{x}_k$ is strictly positive and there is a number $D \ge 1$ such that $\mathbf{x}_k/\mathbf{x}_{k+1} > D$ for $1 \le k \le t-1$ whenever $A\mathbf{x} \ge \mathbf{0}$.

The goal of this section is to partition $\mathbf{x}_1, \dots, \mathbf{x}_t$ into \emph{levels}. 
We partition the variables into sets whose ratios are bounded, while the ratio of variables in different levels can be arbitrarily large. 

\begin{definition}
    A \emph{partition of $1,\dots,t$ into levels} is given by the number of levels $L$, an increasing, surjective map $\mathcal{L} \colon \{1,\dots,t\} \to \{1, \dots,L\}$, and a number $C$ such that for all $j,k \in \{1,\dots,t\}$,
\begin{equation*}
    \mathcal{L}(j) = \mathcal{L}(k) \Longleftrightarrow \mathbf{x}_j/\mathbf{x}_k \in (1/C,C) \text{ whenever } A\mathbf{x} > \mathbf{v}\,.
\end{equation*}
\end{definition}
If $j, k$ are such that $\mathcal{L}(j) < \mathcal{L}(k)$, then we say that $\mathbf{x}_j$ is at a \emph{lower level} than $\mathbf{x}_k$, and if $\mathcal{L}(j) > \mathcal{L}(k)$, $\mathbf{x}_j$ is at a \emph{higher level} than $\mathbf{x}_k$.
We call $\mathcal{L}^{-1}(1)$ the \emph{lowest level} and $\mathcal{L}^{-1}(\ell)$ the \emph{$\ell$th level}.
By abuse of notation, we sometimes refer to an index $k$ being on a certain level and sometimes to an indeterminate $\mathbf{x}_k$ on a certain level.
For the remainder of the paper, we use the notation $L$, $\mathcal{L}$, and $C$ solely for the unique partition into levels.
\begin{proposition}\label{prop:levels}
    Let $A\mathbf{x} > \mathbf{v}$ be a non-thin polyhedron. 
    There is a partition of $1,\dots,t$ into levels, and for this partition, $L$, $\mathcal{L}$, and $C$ can be effectively computed.
\end{proposition}
\begin{proof}
    We inductively construct $\mathcal{L}$ and $C$. 
    Then $L$ is the cardinality of the image of~$\mathcal{L}$.
    To start, set $C = 1$ and $A' = A$.
    Using the decidability of Presburger arithmetic, find the smallest $1 \le \lambda < t$ such that $A_{1:\lambda}'\mathbf{x}' \ge \mathbf{0}$ for some $\mathbf{x}' \in \rat_{> 0}^\lambda$.

    We first show that $\mathbf{x}_1/\mathbf{x}_\lambda$ is bounded whenever $A'_{1:\lambda}(\mathbf{x}_1,\dots,\mathbf{x}_\lambda) \ge \mathbf{0}$.
    Otherwise, let $(\mathbf{y}^{(n)})_{n=0}^\infty \subset \rat_{> 0}^\lambda$ satisfy $A'_{1:\lambda}\mathbf{y}^{(n)} \ge \mathbf{0}$ and $\mathbf{y}_1^{(n)}/\mathbf{y}_\lambda^{(n)} > n$.
    By homogeneity of the inequality $A'\mathbf{y}^{(n)} \ge \mathbf{0}$, we can assume that $\mathbf{y}^{(n)}_1 = 1$. 
    Thus, $\mathbf{y}^{(n)}_k < D^{-(k-1)}$ for $1 \le k \le \lambda$ and $0 < \mathbf{y}^{(n)}_\lambda < 1/n$.
    The Bolzano-Weierstrass theorem gives a limit point $\tilde{\mathbf{y}} \in [0,D^{-(\lambda-1)}]^\lambda$ of  $(\mathbf{y}^{(n)})_{n=0}^\infty$ for which
    $\tilde{\mathbf{y}}_1 = 1$, $\tilde{\mathbf{y}}_\lambda = 0$, and $A_{1:\lambda}'\tilde{\mathbf{y}} \ge \mathbf{0}$.
    As $A'$ has rational entries, there is such a point $\mathbf{y} \in \rat^\lambda$.
    Multiplying $\tilde{\mathbf{y}}$ with an appropriate integer gives an integer-valued vector $\mathbf{y}$ such that $A'_{1:\lambda}\mathbf{y} \ge \mathbf{0}$ and $\mathbf{y}_\lambda = 0$. 
    Thus, for some $1 \le j < \lambda$, $\mathbf{y}_j \ne \mathbf{y}_{j+1} = \cdots = \mathbf{y}_\lambda = 0$ and so $\mathbf{y}_1,\dots,\mathbf{y}_j > 0$ as $\mathbf{y}_k\ge D\mathbf{y}_{k+1} $ for $1 \le k < \lambda$. 
    Thus, $A'_{1:j}(\mathbf{y}_1,\dots,\mathbf{y}_j) \ge \mathbf{0}$, contradicting the minimality of $\lambda$ and thus proving our claim.
    
    Thus, there is a number $C'$ such that $\mathbf{x}_j/\mathbf{x}_k \in (1/C',C')$ for all $1 \le j,k\le \lambda$, and any $C \ge C'$ has this property. 
    To compute a $C$, simply add the inequality $\mathbf{x}_1  \le C\mathbf{x}_\lambda$ to $A'$ to get a matrix $B'$ until $B'(\mathbf{x}_1,\dots,\mathbf{x}_\lambda) \ge \mathbf{0}$ does not hold for all $\mathbf{x}$. 
    Now take $C$ large enough that it has this property for all levels and update $A'$ to the projection of $A'\mathbf{x} \ge \mathbf{0}$ on $\mathbf{x}_\lambda,\dots,\mathbf{x}_t$.

    If no such $\lambda$ exists, try $\lambda = t$, where in the computation of $C$, one takes $\mathbf{v}$ into account. As the original polyhedron $A\mathbf{x} > \mathbf{v}$ was not thin, $A'\mathbf{x} \ge \mathbf{0}$ indeed has a solution.
\end{proof}

Next, we show how to compute a basis for the orthogonal subspaces of the levels.
For a level $1 \le \ell \le L$, let $B_\ell$ denote the matrix $A'_{1:\lambda}$ when computing the $\ell$th level. 
This $B_\ell$ gives the \emph{induced subpolyhedron} $B_\ell(\mathbf{x}_{\lambda},\dots,\mathbf{x}_{\lambda'}) \ge \mathbf{0}$ with a \emph{orthogonal complement of the $\ell$th level} $U_\ell$.
Let $p_\ell$ denote the dimension of $U_\ell$ and $\mathbf{u}_{\ell,1},\dots,\mathbf{u}_{\ell,p_\ell}$ a basis of $U_\ell$, which we can compute by Lemma~\ref{lem:compute basis orth space}.
We observe that each induced subpolyhedron $B_\ell$ has intrinsic dimension at least one. 
Otherwise, $\mathbf{0}$ is the only point in an induced subpolyhedron, which contradicts that the induced subpolyhedron contains a point $\mathbf{x}' > \mathbf{0}$.

Finally, we formally define the bottom family and power variables with small coefficients.

\begin{definition}\label{def:bottom family}
    Assume property~\ref{step 1}. Let $f \ge 1$ and  $F \subseteq \{1,\dots,f\}$ be such that $k \in F$ if and only if 
    \begin{itemize}
        \item $\mathbf{x}_k$ is a pure power variable; or
        \item $\mathbf{x}_k = x_{i,j}z_j^{n_j}$ is a linear power variable with an accompanying pure power variable $\mathbf{x}_{k'} = z_j^{n_j}$ satisfying $k' \le f$.
    \end{itemize}
    Further, assume $f \ge 1$ is the smallest number such that
\begin{enumerate}
    \item $f = t$ or $\mathcal{L}(f) \ne \mathcal{L}(f+1)$; and
    \item the projection of the polyhedron $A_{1:f}(\mathbf{x}_1,\dots,\mathbf{x}_f) \ge \mathbf{0}$ onto the coordinates in $F$ is not thin.
\end{enumerate}
Then the set $\{1,\dots,f\}$ forms the bottom family and $F$ indexes the power variables with small coefficients.
\end{definition}

Forcing $f \ge 1$ ensures that the bottom family is non-empty, and item~1 ensures that the bottom family only contains entire levels.
Item~2 says that there is no linear combination of power variables with small coefficients $\sum_{j=1}^e c_j\mathbf{x}_{k_j}$ that is non-trivial (i.e., not all $c_j$ are zero) whose size is bounded by a multiple of a power variable outside the bottom family. 

\begin{lemma}\label{lem:bottom family computable}
    One can compute the bottom family and power variables with small coefficients, i.e., $f$ and $F$.
\end{lemma}
\begin{proof}
    Take the smallest $1 \le f \le t$ that satisfies the conditions above.
    As $t$ always satisfies these conditions, it remains to show that one can determine whether the two conditions above are satisfied, where only the second condition is non-trivial, which immediately follows by combining Lemmas~\ref{lem:lin programming trick} and~\ref{lem:compute basis orth space}.
\end{proof}

Having shown we can compute the levels, bottom family, etc., we verify that we can assume all properties outlined in Section~\ref{sec:outline} are satisfied.
\begin{lemma}\label{lem:all properties}
    We can reduce to systems where properties~\ref{step 1},~\ref{step 2},~\labelcref{step 3a,step 3b,step 3c,step 3d}, and~\labelcref{step 4a,step 4b,step 4c} are satisfied simultaneously.
\end{lemma}
\begin{proof}
    By Lemma~\ref{lem:property 3a}, property~\ref{step 3a} can always be assumed. Properties~\ref{step 1},~\ref{step 2}, and~\ref{step 3b} can also always be assumed due to Lemmas~\ref{lem:properties 1,2} and~\ref{lem:d_i power of z_(k_i) 3b}. 
    The first is sufficient to compute the levels, bases of the orthogonal subspaces of the levels, bottom family, and power variables with small coefficients due to Proposition~\ref{prop:levels} and Lemma~\ref{lem:bottom family computable}.
    Properties~\labelcref{step 4a,step 4b} can then also be guaranteed by Lemma~\ref{lem:step 4a, 4b}, as well as properties~\ref{step 3d} and~\ref{step 4c} by Lemma~\ref{lem:replacements terminate 4c}.
    Lemma~\ref{lem:step 3c} gives the last property, property~\ref{step 3c}. 
\end{proof}

In the spirit of Proposition~\ref{prop:levels}, we have the following lemma whose proof is in Appendix~\ref{app:levels and bottom family}.

\begin{lemma}\label{lem:Bound for thinness}
    Let $D \in \rel_{>0}$, $e< t$, and $A\mathbf{x} > \mathbf{v}$ be a polyhedron whose internal points satisfy $\mathbf{x}_k > D\mathbf{x}_{k+1}$ for $k = 1,\dots,t-1$. 
    Further, let $\mathbf{c} = (\mathbf{c}_1,\dots,\mathbf{c}_e) \in \intg^e\setminus\{\mathbf{0}\}$ such that $\sum_{i=1}^e \mathbf{c}_i\mathbf{x}_i = 0$ whenever $A_{1:e}(\mathbf{x}_1,\dots,\mathbf{x}_e) \ge \mathbf{0}$. Then there is a computable number $d$ such that $|\sum_{i=1}^e \mathbf{c}_i\mathbf{x}_i| < d\mathbf{x}_{e+1}$ whenever $\mathbf{x} \in\intg_{>0}^t$ satisfies $A\mathbf{x} > \mathbf{v}$. 
\end{lemma}

\subsection{The quantifier elimination argument}\label{sec:Finalising the argument}

In the next section, we finalise our quantifier elimination argument by showing how to eliminate a large number of variables assuming the properties from Section~\ref{sec:outline} (which we can assume by Lemma~\ref{lem:all properties}).
Our method has many moving parts, so it is difficult to outline it directly. 
We start with the following crucial tool that states that there is always a solution as long as certain properties are satisfied.

\begin{theorem}\label{thm:polyhedron no modular constraints}
    Let $A\mathbf{x} > \mathbf{v}$ be a polyhedron such that the entries of $\mathbf{x}$ are power variables that satisfy properties~\ref{step 1},~\ref{step 2}, and~\labelcref{step 4a,step 4b,step 4c}.
    For every $R > 0$, we can find a solution in power variables $\mathbf{x}_k$ such that when $\mathbf{x}_k$ and $\mathbf{x}_{k+1}$ are not on the same level, we have that $\mathbf{x}_k/\mathbf{x}_{k+1} > R$.
\end{theorem}

This theorem is among the most technical parts of the paper, and due to space constraints, its proof is in Appendix~\ref{app:no further constraints}. 
Kronecker's theorem on Diophantine approximation (Theorem~\ref{thm:kronecker dense}) is its main technical tool.

First, let us consider the case where all variables are in the bottom family.
Then all power variables have small coefficients and the cone is not thin, and property~\ref{step 3c} implies that there are no modular constraints with variable modulus. 
The modular constraints with a constant modulus also form no obstacle.
First, we can assume that all are of the form $\mathbf{x}_k \equiv s_k \pmod{b}$. For example, if our modular constraints were $\mathbf{x}_1 + \mathbf{x}_2 \equiv 1 \pmod 4$ and $\mathbf{x}_1 + 2\mathbf{x}_2 \equiv 0 \pmod 6$, then using the Chinese remainder theorem, it is sufficient to search through the pairs $(s_1,s_2) \in \{0,\dots,\operatorname{lcm}(4,6)-1\}^2$ that $\mathbf{x}_1 \equiv s_1 \pmod{24}$ and $\mathbf{x}_2 \equiv s_2 \pmod{24}$.
For $z \in \{\alpha,\beta\}$, $(z^n \bmod{m})_{n=0}^\infty$ is a very well-behaved ultimately periodic sequence.
We either have that $z^n \equiv s \pmod{b}$ has no solutions, one solution, or one can compute $A, B$ such that for all large enough $n$, $z^n \equiv s \pmod{b}$ holds if and only if $n \equiv a' \pmod{b'}$.
Similarly, if $x z^n$ represents a linear power variable, $z^n \bmod{b}$ is known, and $x z^n \equiv s \bmod{b}$ has to hold, then there is either no integer solution for $x$ or $x$ lies in a certain arithmetic progression.
Hence, when there are infinitely many solutions for each $\mathbf{x}_k$ under the modular constraints, we can blow the powers $\alpha$ and $\beta$ up to powers $\alpha^{a'}$ and $\beta^{b'}$ and the linear variables $x_{i,j}$ to $kx_{i,j}$ in one huge case analysis to ensure that the modular constraints with constant modulus always hold.
Then Theorem~\ref{thm:polyhedron no modular constraints} allows us to conclude there always is a solution.

In general, we cannot assume all variables belong to the bottom family. 
Our strategy (for the quantifier elimination) is to find a solution for the power variables outside the bottom family first.
As these are the only power variables that can be a modulus of a modular constraint with variable modulus, this strategy fixes all moduli, turning these constraints into modular constraints with constant modulus on which we can apply the method above.
There are two problems:
\begin{enumerate}
    \item How do we guarantee that for our choice of variables outside the bottom family we indeed have infinitely many modular solutions for each power variable in these new modular constraints?
    \item There could be power variables in the bottom family without small coefficients, and so the cone restricted to the bottom family could be thin. 
    Moreover, as $z_j^{n_j}$ is not in the bottom family, it is fixed, turning this power variable into a usual linear variable.
\end{enumerate}

We tackle the first problem first. 
After dealing with the modular constraints with constant modulus as above, we can apply the following lifting lemma, which is a variation of the lifting-the-exponent lemma. 
Its proof is in Appendix~\ref{app:Finalising the argument}.
\begin{lemma}\label{lem:range modulo}
    Let $\alpha, \beta \ge 2$ be coprime. Then we can compute a number $\ell$ and a set $S \subseteq \{0,\dots,\alpha^\ell-1\}$ such that for all $m \in \nat$, we have that
    \begin{equation*}
        \exists n \in \nat\colon \beta^n \equiv a \pmod{\alpha^m} \:\:\Longleftrightarrow\:\: \exists s \in S \colon s \equiv a \pmod{\alpha^\ell}\,.
    \end{equation*}
\end{lemma}
In this lemma, $\alpha^m$ would be the power variable outside the bottom family, while $\beta$ is in the bottom family. 
Then, instead of studying the modular classes modulo $\alpha^m$, Lemma~\ref{lem:range modulo} allows us to reduce ourselves to considering the modular classes modulo $\alpha^\ell$, which crucially is independent of $m$. 
Thus, like the modular constraints with constant modulus, we can treat the modular constraint with variable modulus by blowing up the exponential bases $\alpha$ and $\beta$. 

For the second problem, we have a fix similar to the modular constraints. 
Instead of restricting the problem to just the variables outside the bottom family, we treat the variables inside the bottom family $\mathbf{x}_1,\dots,\mathbf{x}_b$ at the same time, treating them as usual linear variables $\tilde{\mathbf{x}}_1,\dots,\tilde{\mathbf{x}}_b$ (in a helpful modular class) instead of power variables. 
Then, by Remark~\ref{rem:reduce when equality}, we can eliminate all these linear variables, leaving strictly fewer variables (giving our quantifier-elimination argument).
By identifying each ``thinness'' (roughly said, linearly independent witnesses of thinness) with a power variable without small coefficients inside the bottom family, we can turn them into modular constraints.

For example, if $|2^{n_1} -3^{n_2} + y2^{n_3}|$ is a witness of the cone restricted to the bottom family being thin, $2^{n_1} = \mathbf{x}_1$, $3^{n_2} = \mathbf{x}_2$, and $\mathbf{x}_3 = x2^{n_3}$ are all in the bottom family, we get from the quantifier elimination some variables $\tilde{\mathbf{x}}_1,\dots,\tilde{\mathbf{x}}_3$. 
Then set $x = \tilde{\mathbf{x}}_1- \tilde{\mathbf{x}}_2 +\tilde{\mathbf{x}}_3$ such that it is sufficient that $\tilde{\mathbf{x}}_1- \tilde{\mathbf{x}}_2 \equiv x \pmod{2^{n_3}}$ and $\tilde{\mathbf{x}}_1 > \tilde{\mathbf{x}}_2$.
By the same techniques as in Section~\ref{sec:outline}, we can thus move this thinness into a modular constraint that we treat like the other modular constraints.
In this way, we also deal with the second problem.

The technical implementation of these ideas is in Appendix~\ref{app:Finalising the argument}.
\section{Conclusion}\label{Sec:conclusion}

Having proven Theorem~\ref{thm:main theorem}, which states that $\exists\mathsf{FO}(\intg;<,+,V_\alpha,V_\beta)$ is decidable for coprime $\alpha$ and $\beta$, we should explain our reliance on coprimality as the earlier result~\cite{karimov2025decidability} that established the decidability of $\exists\mathsf{FO}(\intg;<,+,\alpha^\nat,\beta^\nat)$ used the weaker notion of multiplicative independence.
To emphasise this distinction, we formulated our lemmas using multiplicative independence whenever the proofs permitted without significant modifications.
We depended on coprimality twice. 
For property~\ref{step 3c}, we applied two versions of Baker's theorem, Lemmas~\labelcref{lem:Baker helper reduce,lem:Baker to eliminate bottom family modular constraints}; while the latter's proof uses the former, coprimality is used in a single step. 
A similar step occurred in~\cite{karimov2025decidability}, and the authors' solution can be applied \emph{mutatis mutandis}. 
We omitted this generalisation as we do not know how to circumvent our second use of coprimality.
It occurs in Lemma~\ref{lem:range modulo} and its subsequent use in the proof of the main result has no counterpart in~\cite{karimov2025decidability}.

Lemma~\ref{lem:range modulo} describes the structure of $\{\beta^n \bmod{\alpha^m}: n \in\nat\}$ for a given number $m$. 
Coprimality is equivalent to $\beta$ being invertible modulo powers of $\alpha$, and so $s \equiv \beta^n \pmod{\alpha^m}$ solely depends on $s \bmod \alpha^\ell$ for a computable $\ell$ independent of $m$. 
Thus, the set $\{\beta^n \bmod{\alpha^m}: n \in\nat\}$ is well-behaved: $\#\{\beta^n \bmod{\alpha^m}: n \in\nat\}/\alpha^m$ is the same for all $m \ge \ell$.
However, for non-coprime bases, matters are more complicated. 
Take $\alpha = 4$ and $\beta = 6$.
The set $\{6^n \bmod{4^m}: n \in\nat\}$ consists of $2m+1$ numbers\footnote{This follows because $6^{2m} \equiv 0 \pmod{4^m}$, $6^{2m-1} \not\equiv 0 \pmod{4^m}$, and the number of factors of two in $(6^n)_{n=0}^\infty$ is strictly increasing, showing that $1, 6, 6^2,\dots,6^{2m}$ are all distinct modulo $4^m$.}, which gives a ratio $(2m+1)/4^m = \#\{6^n \bmod{4^m}: n \in\nat\}/4^m$.
For a given $s$, the equation $6^n \equiv s \pmod{4^m}$ either has no solutions, a single solution, or $s \equiv 0 \pmod{4^m}$ and all $n \ge 2m$ suffice.

Similarly, for a given $m$, $\{4^n \bmod{6^m}:n \ge 2m\}$ consists of an exceptional set (growing linearly in $m$) of numbers reached by a single number and a set of numbers reached infinitely often.
However, these sets grow much more slowly than $6^m$.
We would expect that new ideas are necessary to navigate these issues, possibly using a $p$-adic version of Baker's theorem.

We expect that for certain multiplicatively independent $\alpha$ and $\beta$ that are not coprime, $\exists\mathsf{FO}(\nat; +, V_\alpha, \beta^\nat)$ can be shown to be decidable with these techniques, but that the general case (i.e., all multiplicatively independent $\alpha$ and $\beta$) remains out of reach.

Beyond extending the result from coprime to multiplicatively independent bases, we also want to discuss some other directions.
One natural question is whether one can add more than two B\"uchi predicates, or, as an easier version, three power predicates. 
For example, is $\exists\mathsf{FO}(\intg;<,+,2^\nat,3^\nat,5^\nat)$ decidable? In~\cite{berthe2024decidability}, the authors needed Schanuel's conjecture to deal with problems of the form $\exists n_2,n_3,n_5 \colon d_1z_1^{n_{z_1}} < \cdots < d_kz_k^{n_{z_k}}$ where $z_i \in \{2,3,5\}$ and each $d_i$ is a power of~$z_i$.
We avoid this problem as we apply Kronecker's theorem on a one-dimensional torus and only need that $\log(\alpha)/\log(\beta)$ is irrational. 
For these three powers $2$, $3$, and $5$ one needs that $1/\log(2)$, $1/\log(3)$, and $1/\log(5)$ are linearly independent, for which Schanuel's conjecture can be used. 
It is also unclear how the Baker-type arguments generalise to three different bases.
Hence, this direction seems challenging.

Another natural generalization is $\exists\mathsf{FO}(\intg;<,+,P_\alpha,P_\beta)$, where $P_z$ is a $z$-regular set for $z \in \{\alpha,\beta\}$.
That is, $P_z$ is a predicate definable in $\mathsf{FO}(\intg;<,+,V_z)$.
As the existential fragment of B\"uchi arithmetic is strictly less expressive than its full theory, Theorem~\ref{thm:main theorem} does not apply to $\exists\mathsf{FO}(\intg;<,+,P_\alpha,P_\beta)$. 
In fact, when $P_\alpha$ and $P_\beta$ are unary and are interpreted as subsets of $\nat$, determining $P_\alpha \cap P_\beta = \emptyset$, $|P_\alpha \cap P_\beta| = \infty$, or $P_\alpha \subset P_\beta$ all seem difficult.
For example, setting $\alpha = 2$, $P_\alpha = \{2^n \colon n \ge 9\}$, $\beta = 3$, and $P_\beta$ as the set of numbers without a $2$ in their ternary expansion, one could ask to determine whether $P_\alpha \cap P_\beta = \emptyset$. Erd\H{o}s (see, e.g.,~\cite{Erdos1979Unconventional}) conjectured that this intersection is indeed empty. 
However, he wrote that ``As far as I can see, there is no method at our disposal to attack this conjecture.''
Moreover, Tao~\cite{tao2011collatz} wrote that this conjecture ``is still a fair distance beyond what one can do with current technology.''
\bibliography{bib}

\appendix

\section{Proofs from Section~\ref{sec:prelims}}\label{app:Baker}
\begin{proof}[Proof of Lemma~\ref{lem:thinness computable}]
    If the lemma holds for $\sim$ equal to $\ge$, it also holds for $\sim$ equal to $>$.
    For a given polyhedron $A\mathbf{x} \ge \mathbf{v}$ of ambient dimension $t$, let $U$ be the orthogonal complement of its recession cone.
    By the proof of~\cite[Corollary 7.1b]{schrijver1998theory}, compute a bounded polyhedron $B\mathbf{y} \ge \mathbf{w}$ such that $A\mathbf{x} \ge \mathbf{v}$ if and only if $\mathbf{x} = \mathbf{x}'+\mathbf{y}$ for some $\mathbf{x}'$ and $\mathbf{y}$ satisfying $A\mathbf{x}' \ge \mathbf{0}$ and $B\mathbf{y} \ge \mathbf{w}$.

    If $\dim(U) = 0$, let $\mathbf{u} \in\rat^t \setminus \{\mathbf{0}\}$ be outside $U$ and $\mathbf{y}$ satisfy $B\mathbf{y} \ge \mathbf{w}$.
    Then some $\mathbf{x}'$ in the recession cone satisfies $\mathbf{u} \cdot \mathbf{x}' \ne 0$.
    The cited corollary implies that for any $a \ge 0$, $\mathbf{x} \coloneqq a\mathbf{x}' + \mathbf{y}$ satisfies $A\mathbf{x} \ge\mathbf{v}$ while $|\mathbf{u}\cdot(a\mathbf{x}'+\mathbf{y})| \ge a|\mathbf{u}\cdot\mathbf{x}'| - |\mathbf{u}\cdot\mathbf{y}|$ is unbounded.
    Hence, the polyhedron $A\mathbf{x} \ge \mathbf{v}$ is not thin.

    If $\dim(U) > 0$, compute $\mathbf{u} \in U\setminus\{\mathbf{0}\}$ and a number $e$ such that $|\mathbf{y}_i| < e$ for every $\mathbf{y}$ in the bounded polyhedron $B\mathbf{y} \ge \mathbf{w}$.
    Then, writing $\mathbf{x}$ in the polyhedron $A\mathbf{x} \ge \mathbf{v}$ as $\mathbf{x} = \mathbf{x}'+\mathbf{y}$ as in the cited corollary, we have that $ \left|\mathbf{u}\cdot\mathbf{x} \right| \le te\max_{1 \le i \le t}|\mathbf{u}_i|$ as $\mathbf{u}$ is orthogonal to $\mathbf{x}'$.
    Then take $f = te\max_{1 \le i \le t}|\mathbf{u}_i|$.
\end{proof}
\begin{proof}[Proof of Theorem~\ref{thm:kronecker dense}]
    Let $\theta = \log\alpha/\log\beta$. Then $\theta$ is irrational, as otherwise, if $\theta  =p/q$ is rational, $\beta^p = \alpha^q$, contradicting multiplicative independence.
    Thus, $\{n\theta - m : m \in \intg_{>M}, n \in \intg_{>N}\}$, and so $\{n\log\alpha - m\log\beta : m \in \intg_{>M}, n \in \intg_{>N}\}$, is dense in $\rel$.
    Applying the continuous and surjective exponential function $\exp:\rel \to \rel_{>0}$ gives the theorem.
\end{proof}
\begin{proof}[Proof of Lemma~\ref{lem::soda observation}]
    Otherwise, $\mathcal{I}_2 = \mathcal{I}_3 = \emptyset$, trivially fulfilling the right-hand side of the if and only if condition. Hence,~\eqref{eq:Soda 2025} holds for every $(k_1,\dots,k_r)$, and, in particular, for $(k_1,\dots,k_r) = (0,\dots,0)$ and for $(k_1,\dots,k_r) = (0,\dots,0,1,0,\dots,0)$ with the $1$ at the $i$th position. 
    Subtracting these two equations gives that $c_i(z_i^0 - z_i^1) = 0$, which is impossible as $z_i \ge 2$ and $c_i \ne 0$.
\end{proof}
\begin{proof}[Proof of Lemma~\ref{lem:Baker helper reduce}]
    If $x_\alpha = 0$, then item 4 gives $|x_\beta| \ge 1$ and so $\beta^{k_\beta} < C_3Z^m$. 
    Now item 3 gives $\beta^{\ell_\beta - p_\beta(\log\ell_\beta)} < C_3Z^m$.
    If $Z = \beta$, we thus have that $\ell_Z - m$ is effectively bounded. 
    If $Z = \alpha$, item $1$ implies $C_1^{-1}\alpha^{\ell_\alpha-\log_\alpha(\beta) p_\beta(\log\ell_\beta)} < \beta^{\ell_\beta - p_\beta(\log\ell_\beta)}< C_3Z^m $ and again $\ell_Z - m$ is effectively bounded.
    Thus, we can assume that $x_\alpha \ne 0$, and by symmetry that $x_\beta \ne 0$.

    Again, by symmetry, we can assume that $Z = \alpha$. By item 1, we have that $\ell_\alpha\log\alpha - \ell_\beta\log\beta \in (-\log C_1,\log C_1)$ and so we can compute $N, s, t$ such that $\ell_\beta < s \ell_\alpha + t$ whenever $\ell_\alpha > N$.
    Hence, we can compute $s', t'$ such that
    \begin{equation}\label{eq:Baker help max}
        \log\max(3,k_\alpha,k_\beta) \le \log\max(3,\ell_\alpha,\ell_\beta) < s'\log\ell_\alpha + t'\,.
    \end{equation}
    Similarly, we can also compute a polynomial $q$ such that 
    \begin{equation}\label{eq:baker max poly}
        \max_{z \in \{\alpha, \beta\}}p_z(\log\ell_z)\log z < q(\log\ell_\alpha)
    \end{equation}
    Dividing item 4 by $x_\alpha \alpha^{k_\alpha} \ne 0$ and using that $\frac{1}{|x_\alpha|} \le 1$ as $x_\alpha \in \intg_{\ne 0}$ gives
    \begin{equation*}
        0< \left|\frac{-x_\beta}{x_\alpha}\frac{\beta^{k_\beta}}{\alpha^{k_\alpha}}-1\right| < C_3\frac{1}{|x_\alpha|}\alpha^{m-k_\alpha} \le C_3\alpha^{m-k_\alpha}\,.
    \end{equation*}
    We apply Theorem~\ref{thm:baker} with $c = -x_\beta/x_\alpha$. 
    Then~\eqref{eq:baker max poly} and items 2 and 3 give
    \begin{align*}
        h(c) &= \max_{z \in \{\alpha, \beta\}}\log|x_z| \\
        &< \max_{z \in \{\alpha, \beta\}}\log\left|C_2 z^{\ell_z-k_z}\right| \\
        &< \log|C_2| + \max_{z \in \{\alpha, \beta\}}p_z(\log\ell_z)\log z \\
        &< \log|C_2| + q(\log\ell_\alpha)\,.
    \end{align*}
    Hence Theorem~\ref{thm:baker} implies that there is a computable constant $C$ such that 
    \begin{equation*}
        C_3\alpha^{m-k_\alpha} > \left|\frac{x_\beta}{x_\alpha}\frac{\beta^{k_\beta}}{\alpha^{k_\alpha}}-1\right| > \max(3,k_\alpha,k_\beta)^{-C(\log|C_2| + q(\log \ell_\alpha))}\,.
    \end{equation*}
    Taking logarithms and applying~\eqref{eq:Baker help max} gives that
       \begin{equation*}
         (\ell_\alpha-m)\log\alpha-\log C_3 < C(\log|C_2| + q(\log\ell_\alpha))(s'\log\ell_\alpha + t')\,.
    \end{equation*}
    As the constants and polynomial are explicit, the polynomial $p$ is computable.
\end{proof}

\begin{proof}[Proof of Lemma~\ref{lem:Baker to eliminate bottom family modular constraints}]
    Our first goal is to inductively compute polynomials $q_1,\dots,q_r$ such that for $1 \le i \le r$, we have that $\ell_{z_i} - n_i < q_i(\log\ell_{z_i})$ implies that 
    \begin{equation*}
        \big|x_\alpha \alpha^{k_\alpha} + x_\beta \beta^{k_\beta}+c_1z_1^{n_1} + \cdots + c_{i-1} z_{i-1}^{n_{i-1}}\big| > |c_i|z_i^{n_i} + \cdots + |c_r|z_r^{n_r}+|c|\,.
    \end{equation*}
    This follows from repeatedly applying Lemma~\ref{lem:Baker helper reduce}.
    For this, we gather the powers of $\alpha$ and $\beta$ on the left-hand side as $x_\alpha' \alpha^{m_\alpha} +x_\beta' \beta^{m_\beta}$.
    Let $z \in \{\alpha,\beta\}$ and set $z^{m_z}$ to be smallest $z_1^{n_1},\dots,z_{i-1}^{n_{i-1}},z^{k_z}$ with exponential base $z$.
    Then $|x_z'| < C_2'z^{\ell_z - m_z}$ for some computable constant~$C_2'$, and $\ell_z - m_z < q_i(\log\ell_z)$ when $z^{m_z} = z_i^{n_i}$ and $\ell_z - m_z < p_z(\log\ell_z)$ when $z^{m_z} = z^{k_z}$.
    We bound $|c_i|z_i^{n_i} + \cdots + |c_r|z_r^{n_r}+|c|$ with $C_3z_i^{n_i} \coloneqq (|c_i|+ \cdots + |c_r|+|c|)z_i^{n_i}$.
    Finally, item 5 implies that $x_\alpha' \alpha^{m_\alpha} +x_\beta' \beta^{m_\beta} \ne 0$, giving item 4 of Lemma~\ref{lem:Baker helper reduce}, showing that the lemma can be applied and that $q_1,\dots,q_r$ can be computed.

    We claim that we can reduce to the case $r = 0$. 
    We apply the same argument as before. For $z \in \{\alpha,\beta\}$, let $z^{m_z}$ be the smallest $z_1^{n_1},\dots,z_r^{n_r},z^{k_z}$ with exponential base $z$.
    We write $x_\alpha' \alpha^{m_\alpha} +x_\beta' \beta^{m_\beta}$ for $x_\alpha \alpha^{k_\alpha} + x_\beta \beta^{k_\beta}+c_1z_1^{n_1} + \cdots + c_rz_r^{n_r}$ by gathering all the powers of $z^{n_z}$, and $p'_z = q_i$ when $z^{m_z} = z_i^{n_i}$ and $p'_z = p_z$ when $z^{m_z} = z^{k_z}$.
    We can compute a constant $C_2'$ such that $|x_z'| < C_2'z^{\ell_z-m_z}$ and $\ell_z - m_z < p_z'(\log\ell_z)$ when items 1--5 hold.
    Lastly, in this new setting, item 5 holds, as $x_\alpha'$ and $x_\beta'$ are non-empty proper subsums of the multiset $\{x_\alpha \alpha^{k_\alpha}, x_\beta \beta^{k_\beta},c_1z_1^{n_1}, \ldots, c_rz_r^{n_r}\}$.
    Then the hypothesis of this lemma is still satisfied, and we have reduced to the case $r = 0$.

    If we bound $\ell_\beta$ from above, item 2 bounds $\ell_\alpha$ from above, and then item 1 bounds $k_\alpha,k_\beta$ from above.
    Thus, it is sufficient to bound $\ell_\beta$ from above.

    Assume $c = 0$.
    If $x_\alpha \alpha^{k_\alpha} + x_\beta \beta^{k_\beta} = 0$, $x_\alpha$ is a non-zero multiple of $\beta^{k_\beta}$ as $\alpha$ and $\beta$ are coprime. 
    Thus, by items 2--4, we have that 
    \begin{equation*}
        \beta^{\ell_\beta - p_\beta(\log\ell_\beta)} < \beta^{k_\beta} \le |x_\alpha| < C_2 \alpha^{\ell_\alpha - k_\alpha} < C_2 \alpha^{p_{\alpha}(\log\ell_\alpha)} < C_2 \alpha^{p_{\alpha}(\log(C_1\ell_\beta))}\,.
    \end{equation*}
    Taking logarithms gives that
    \begin{equation*}
        (\ell_\beta - p_\beta(\log\ell_\beta))\log\beta< \log C_2+(p_{\alpha}(\log(C_1\ell_\beta)))\log\alpha\,.
    \end{equation*}
    This allows us to bound $\ell_\beta$ from above as the left-hand side grows linearly while the right-hand side grows logarithmically in $\ell_\beta$.

    If $c \ne 0$, introduce a new variable $m$ such that
    \begin{equation*}
        x_\alpha \alpha^{\ell_\alpha} + x_\beta \beta^{\ell_\beta}+c_1\beta^m = 0\,.
    \end{equation*}
    Then apply Lemma~\ref{lem:Baker helper reduce} with $Z = \beta$ and $C_3 = |c_1|+1$ to obtain a polynomial $p$ such that $\ell_\beta - m < p(\log\ell_\beta)$.
    Setting $m = 0$ we have that $\ell_\beta < p(\log\ell_\beta)$, bounding $\ell_\beta$ from above.
\end{proof}
\section{Proofs from Section~\ref{sec:outline}}\label{app:properties}
\begin{proof}[Proof of Lemma~\ref{lem:properties 1,2}]
    If $\mathbf{x}_{k'} > \mathbf{x}_k$, then $x_{i,j}$ is an integer such that $0 < x_{i,j}z_j^{n_j} < z_j^{n_j}$.
    As the above is impossible, we resolve this system immediately. If $t = 0$, the problem reduces to a quantifier-free sentence whose truth is straightforwardly verified.
    For property~\ref{step 2}, by Lemma~\ref{lem:thinness computable}, we can compute a witness $\mathbf{u} \in \rat^t$ and $f \in \rat_{\ge 0}$ such that $|\mathbf{u}\cdot\mathbf{x}| < f$ whenever $A\mathbf{x} >\mathbf{v}$.
    By multiplying with an appropriate (and computable) constant $m \in \intg_{>0}$, we can assume that $\mathbf{u}$ and $f$ are integer-valued. 
    Now one of $\mathbf{u}\cdot\mathbf{x} = c$ for $c \in\{-f+1,\dots,f-1\}$ has to hold.
    Split into these systems, and apply Remark~\ref{rem:reduce when equality} to reduce to a system with at most $t-1$ variables.
\end{proof}

\begin{proof}[Proof of Lemma~\ref{lem:d_i power of z_(k_i) 3b}]
    Assume property~\ref{step 3b} fails for $d_iz_{k_i}^{n_{k_i}}$. 
    Factor $d_i$ as $d_i = ef$ such that $\gcd(e,z_{k_i}) = 1$ and every prime divisor of $f$ divides $z_{k_i}$. 
    Then $f$ divides a (computable) power of $z_{k_i}$, say $z_{k_i}^m$.
    By the Chinese remainder theorem,
    \begin{align*}
        &\sum_{k=1}^tc_{i,k}\mathbf{x}_k\equiv c_i \pmod{d_i z_{k_i}^{n_{k_i}}} \\\Longleftrightarrow &\sum_{k=1}^tc_{i,k}\mathbf{x}_k\equiv c_i \pmod{e} \land \sum_{k=1}^tc_{i,k}\mathbf{x}_k\equiv c_i \pmod{f z_{k_i}^{n_{k_i}}}\,.
    \end{align*}
    Add the modular constraint with modulus $e$ to the modular constraints with constant modulus and replace $d_i$ with $f$ in the original constraint. 
    Multiplying this constraint by $z_{k_i}^m/f$ gives an equivalent modular constraint where the new $d_i$ is a power of $z_{k_i}$.
    This process does not alter any other modular constraint or the polyhedron and preserves properties~\ref{step 3a},~\ref{step 3d}, and~\ref{step 4c}.
    Apply this process to the modular constraints with variable modulus one by one. 
\end{proof}

\begin{proof}[Proof of Lemma~\ref{lem:step 3c}]
    Assume $\sum_{k=1}^t c_{i,k}\mathbf{x}_k \equiv c_i \pmod{d_iz_{k_i}^{n_{k_i}}}$ with $1 \le i \le s'$ violates this property.
    Let $x \in\intg$, $1 \le k'\le t$ such that $\sum_{k=1}^t c_{i,k}\mathbf{x}_k - c_i = xd_iz_{k_i}^{n_{k_i}}$ and $\mathbf{x}_{k'} = z_{k_i}^{n_{k_i}}$.

    First, assume that $c_{i,k} = 0$ when $\mathcal{L}(\mathbf{x}_k) < \mathcal{L}(\mathbf{x}_{k'})$ and $\mathbf{x}_k$ has exponential basis different than $z_{k_i}$.
    Then, $c_{i,1} = \cdots = c_{i,k'-1} = 0$ by property~\ref{step 3d}, and so we can compute a constant $C_1$ such that $\big|\sum_{k=1}^t c_{i,k}\mathbf{x}_k - c_i\big| < C_1\mathbf{x}_{k'}$, leaving only finitely many options for $x$. 
    These cases either yield a trivial equation with zero coefficients or allow for a reduction in the number of variables by Remark~\ref{rem:reduce when equality}.
    Due to this claim, we can also assume that $\mathbf{x}_{k'}$ is not on the first level and thus the bottom family consists of at least two levels.

    We claim that, for $z \in \{\alpha,\beta\}$, there is a pure power variable on the first level $\mathcal{L}^{-1}(1) = \{1,\dots,\lambda\}$ with exponential basis $z$.
    Otherwise, we argue, the first level $\{1,\dots,\lambda\}$ would satisfy the two conditions to be the bottom family, and by minimality, has to be the bottom family.
    This family satisfies the first condition of solely containing entire levels, and there is at most one power variable with a small coefficient on the first level by properties~\labelcref{step 4a,step 4b}.
    Thus, $|F| = 1$.
    In the notation of the second condition, $B\mathbf{y} \ge \mathbf{0}$ is a thin polyhedron of ambient dimension at most $1$, forcing that $\mathbf{y} = \mathbf{0}$. 
    This contradicts the definition of a level as $A_{1:\lambda}(\mathbf{x}_1,\dots,\mathbf{x}_\lambda) \ge \mathbf{0}$ for  $\mathbf{x}_1,\dots,\mathbf{x}_\lambda > 0$.
    Our claim follows.
    
    For $z \in \{\alpha,\beta\}$, let $z^{\kappa_z}$ and $z^{\ell_z}$ be the smallest and largest pure power variable with exponential base $z$ in the bottom family, respectively, which exist by the previous claim.
    Thus, $z^{\kappa_z} \le z^{\ell_z}$ and $z^{\ell_z}$ is on the first level.
    
    Assume that the bottom family consists of $\mathbf{x}_1,\dots,\mathbf{x}_f$.
    Thus, all power variables with small coefficients $\mathbf{x}_k$ satisfy $k\le f$ and if $z_j^{n_j}$ denotes their exponential part $\kappa_{z_j} \le n_j \le \ell_{z_j}$. If $z_{k_i} = z_j$, we have that $\kappa_{z_j} \le n_{k_i} \le \ell_{z_j}$.
    Hence, we write 
    \begin{equation}\label{eq:split into bottom family vs pure powers}
        x_\alpha \alpha^{\kappa_\alpha} + x_\beta \beta^{\kappa_\beta} = xd_iz_{k_i}^{n_{k_i}} - \sum_{k = 1}^f c_{i,k}\mathbf{x}_k = \sum_{k = f+1}^t c_{i,k}\mathbf{x}_k - c_i\,,
    \end{equation}
    where $x_{\alpha}$ and $x_\beta$ are integers. 
    We verify the conditions of Lemma~\ref{lem:Baker to eliminate bottom family modular constraints}.
    We do not include the terms $c_{i,k}\mathbf{x}_k$ for $k = f+1,\dots,t$ where $c_{i,k} = 0$, and $\mathbf{x}_k$ is a pure power variable when $c_{i,k} \ne 0$. 
    This claim follows quickly. As $\mathbf{x}_{k'}$ is in the bottom family, $k'\le f < k$.
    Thus, $\mathbf{x}_{k}$ has a smaller exponential part than $\mathbf{x}_{k'} = z_{k_i}^{n_{k_i}}$.
    Therefore, $\mathbf{x}_k$ is a pure power variable by property~\ref{step 3a}. 

    For condition~\ref{Baker cond 1}, the first part follows from the strict ordering $\mathbf{x}_1 > \dots > \mathbf{x}_t$ and that $\alpha^{\kappa_\alpha}$ and $\beta^{\kappa_\beta}$ being in the bottom family implies that their corresponding power variables have index at most $f$. The second part follows from the construction of $\kappa_\alpha, \kappa_\beta, \ell_\alpha, \ell_\beta$. 
    Condition~\ref{Baker cond 2} follows from $\alpha^{\ell_\alpha}$ and $\beta^{\ell_\beta}$ being on the first level, and we take $C_1 \coloneqq C$ from the level construction.

    For condition~\ref{Baker cond 3}, fix $z \in \{\alpha,\beta\}$.
    Then, for $1 \le k \le f$, power variables $\mathbf{x}_k$ with exponential base $z$ contain a factor $z^{\kappa_z}$ and are at most $Cz^{\ell_z}$. 
    Hence, the total that the $\mathbf{x}_k$ can contribute to $x_\alpha$ is at most $C\sum_{k=1}^f|c_{i,k}|$.
    For $x$ itself, by the same argument we have that $|x| < \frac{1}{d_i}C(|c_i|+\sum_{k=1}^t|c_{i,k}|)z^{\ell_z-\kappa_z}$.
    This observation allows us to determine a valid $C_2$.

    Condition~\ref{Baker cond 4} needs more work. 
    We apply induction on the pure power variables in the bottom family, starting with $\kappa_z' \coloneqq \ell_z$ for $z \in\{\alpha,\beta\}$, and decreasing one of $\kappa_\alpha'$ and $\kappa_\beta'$ at each step according to the ordering $\mathbf{x}_1 > \cdots > \mathbf{x}_f$ until the condition holds for $\kappa_\alpha'= \kappa_\alpha$ and $\kappa_\beta'= \kappa_\beta$. 
    For the induction basis, we take $p_z' = 1$ for $z \in\{\alpha,\beta\}$ as $\ell_z - \kappa_z'= 0 < 1$.
    For the induction step, assume that $(k'_\alpha,k'_\beta) \ne (\kappa_\alpha,\kappa_\beta)$ (as we are done otherwise) and assume that for $z \in\{\alpha,\beta\}$ we have computed a polynomial $p_z$ such that $\ell_z - \kappa_z'< p_z(\log\ell_z)$ holds whenever $A\mathbf{x} >\mathbf{v}$. 
    Compute $k$ and $Z \in \{\alpha,\beta\}$ such that $\mathbf{x}_k = Z^m$ is the largest pure power variable smaller than $\alpha^{\kappa_\alpha'}$ and $\beta^{\kappa_\beta'}$, which by construction is in the bottom family.
    If $z\in\{\alpha,\beta\}$ and $\mathbf{x}_k$ is on the same level as $z^{\kappa_z'}$, then $z \ne Z$ by property~\ref{step 4a}. 
    In that case, $\ell_Z - m$ can be bounded easily using $\ell_z - \kappa_z'< p_z'(\log\ell_z)$ and $z^{\kappa'_z} /Z^m, \alpha^{\ell_\alpha}/\beta^{\ell_\beta} \in (1/C, C)$.
    So we can assume that $\mathbf{x}_k = Z^m$ is smaller and not on the same level as $\alpha^{\kappa_\alpha'}$ and $\beta^{\kappa_\beta'}$ and so on a higher level. 
    Let $f' = \min\mathcal{L}^{-1}(k)-1$ and so $\alpha^{\kappa'_\alpha}$ and $\beta^{\kappa'_\beta}$ are among $\mathbf{x}_1,\dots,\mathbf{x}_{f'}$.
    Then compute $F' \subseteq \{1,\dots,f'\}$ such that $j\in F'$ if and only if $\mathbf{x}_j$ is a pure power variable or a linear power variable whose exponential part is among $\mathbf{x}_1,\dots,\mathbf{x}_{f'}$. 
    Thus $F'$ is a strict subset of the power variables with small coefficients and by the minimality of the bottom family, the second condition of Definition~\ref{def:bottom family} has to be false, i.e., the projection of the polyhedron $A_{1:f'}(\mathbf{x}_1,\dots,\mathbf{x}_{f'}) \ge \mathbf{0}$ on the indices of $F'$ is thin. 
    Using Lemma~\ref{lem:Bound for thinness}, compute a non-zero witness $\mathbf{u} \in \rat^{|F'|}$, i.e., $\sum_{j \in F'} \mathbf{u}_j \mathbf{x}_{j} = 0$ whenever $A_{1:f'}(\mathbf{x}_1,\dots,\mathbf{x}_{f'}) \ge \mathbf{0}$, and a number $d$ such that $\sum_{j \in F'} \mathbf{u}_j \mathbf{x}_{j} \le d\mathbf{x}_{f'+1}$ whenever $A\mathbf{x} > \mathbf{v}$.
    Now, for $z\in\{\alpha,\beta\}$, let $x'_zz^{\kappa'_z}$ be the restriction of $\sum_{j \in F'}\mathbf{x}_j \mathbf{u}_j$ to the variables with exponential base $z$.
    This sum is non-trivial by property~\ref{step 4c}.
    Thus, as $k$ is on the same level as $f'+1$, $\mathbf{x}_{f'+1} < C\mathbf{x}_k = CZ^m$, $|x'_\alpha \alpha^{\kappa'_\alpha} + x'_\beta \beta^{\kappa'_\beta}| < dCZ^m$.
    By splitting into the case where $x'_\alpha \alpha^{\kappa'_\alpha} + x'_\beta \beta^{\kappa'_\beta} = 0$ (which is a non-trivial equality that lets us reduce to a system with fewer variables by Remark~\ref{rem:reduce when equality}), we can assume that $x'_\alpha \alpha^{\kappa'_\alpha} + x'_\beta \beta^{\kappa'_\beta} \ne 0$.
    For $z \in \{\alpha,\beta\}$, using that $z^{\ell_z}$ is on the first level, compute a constant $C_2$ such that $|x'_z| < C_2z^{\ell_z-\kappa'_z}$.
    Hence, Lemma~\ref{lem:Baker helper reduce} lets us compute a polynomial $p$ such that $\ell_Z - m < p(\log\ell_Z)$.
    Thus, by induction, we can compute the polynomials $p_\alpha$ and $p_\beta$ for condition~\ref{Baker cond 4}.
    
    Finally, for condition~\ref{Baker cond 5}, we again split into subcases, the original and those where we add an equality coming from a proper subsum of the multiset $\{x_\alpha \alpha^{\kappa_\alpha}, x_\beta \beta^{\kappa_\beta},c_1z_1^{n_1}, \ldots, c_rz_r^{n_r}\}$ being zero. 
    The first observation of this proof dealt with the case where this subset equals $\{x_\alpha \alpha^{\kappa_\alpha}\}$ or $\{x_\beta \beta^{\kappa_\beta}\}$.
    In all other cases, we add this non-trivial equation, leading to a reduction in the number of variables by Remark~\ref{rem:reduce when equality}.
    Thus, condition~\ref{Baker cond 5} also holds.
    
    Lemma~\ref{lem:Baker to eliminate bottom family modular constraints} upper bounds $n_\alpha$ and $n_\beta$, letting us reduce to systems with strictly fewer power variables. 
\end{proof}

\begin{proof}[Proof of Lemma~\ref{lem:step 4a, 4b}]
    For property~\ref{step 4a}, assume otherwise: Say $\mathbf{x}_k > \mathbf{x}_{k'}$ violates this property. Then, $\mathcal{L}(k) = \mathcal{L}(k')$ and so $z_j^{n_j-n_{j'}}=\mathbf{x}_k/\mathbf{x}_{k'}\in (1, C)$. 
    Thus, there are finitely many, effectively computable possibilities for $n_j - n_{j'}$, and we can make the natural substitutions to reach systems with fewer variables.
    Similarly, for property~\ref{step 4b}, we have otherwise that $(x_{i,j}z_j^{n_j})/z_j^{n_j} = x_{i,j} \in (1/C, C)$, and again we split into finitely many cases with fewer variables by considering the options $x_{i,j} = 1,\dots,\lfloor C\rfloor$.
\end{proof}

\begin{proof}[Proof of Lemma~\ref{lem:replacements terminate 4c}]
    We use the method from the example.
    Assume that $\mathbf{x}_\lambda,\dots,\mathbf{x}_{\lambda'}$ forms the $\ell$th level and that property~\ref{step 4c} holds for all lower levels. 
    Using standard linear algebra, compute a basis for $U_\ell$ and a non-zero, integer-valued element $(c_\lambda,\dots,c_{\lambda'})$ of $U_\ell$ such that $c_i = 0$ for all power variables $\mathbf{x}_i$ of a certain exponential base.
    By the construction of the levels, at least two $c_i$ are non-zero, of which, by property~\ref{step 4a}, at least one is linear. 
    Let $\lambda \le k \le \lambda'$ be such that $c_k \ne 0$, $\mathbf{x}_k = x_{i,j}z^{n_j}$ is a linear power variable, and $\mathbf{x}_k$ has minimal exponential part among all such power variables. 
    Then introduce a new variable $\mathbf{x}'_k = x'z_j^{n_j} \coloneqq \sum_{i=\lambda}^{\lambda'} c_i\mathbf{x}_i$ which is indeed a multiple of $z_j^{n_j}$.
    Thus, $x' \in \intg$. 
    We split into cases to allow us to assume that $x'> 0$. 
    Using Lemma~\ref{lem:Bound for thinness}, there is a number $d$ such that $|\sum_{i=\lambda}^{\lambda'} c_i\mathbf{x}_i| < d\mathbf{x}_{\lambda'+1}$ whenever $A\mathbf{x} > \mathbf{v}$. 
    Thus, $\mathbf{x}_k' < d\mathbf{x}_{\lambda'+1}$.
    
    As in the example, we use this equality to eliminate $\mathbf{x}_k$ from the problem and add the corresponding modular constraint and inequality, as $\mathbf{x}_k$ is a strictly positive integer. 
    It rests to insert $\mathbf{x}_k'$ into $\mathbf{x}_1 > \cdots > \mathbf{x}_{k-1} > \mathbf{x}_{k+1} > \cdots > \mathbf{x}_t > 0$. 
    If $\mathbf{x}'_k$ equals some $\mathbf{x}_{k'}$, then Remark~\ref{rem:reduce when equality} suggests we can reduce to finitely many systems with at most $t-1$ power variables, so we can insert it strictly and split into cases.
    Redo the computation of the levels. 
    The main observation is the following.
    If two variables $\mathbf{x}_{j_1}$ and $\mathbf{x}_{j_2}$ were on the same level, they remain on the same level as their ratio remains bounded.

    If we insert $\mathbf{x}_k'$ later than $\mathbf{x}_{\lambda'+1}$, i.e., $\mathbf{x}_k' < \mathbf{x}_{\lambda'+1}$, then the computation of the first $\ell-1$ levels remains unchanged, and these levels still satisfy property~\ref{step 4c}.
    The $\ell$th level consists of $\{\lambda,\dots,\lambda'\}\setminus\{k\}$ and is closer to property~\ref{step 4c}: Its ambient dimension decreases while its intrinsic dimension remains unchanged.
    If we insert $\mathbf{x}_k'$ such that $\mathbf{x}_j < \mathbf{x}_k'$ for $j < \lambda'+1$, then $\mathbf{x}_k'$, $\mathbf{x}_j$ and $\mathbf{x}_{\lambda'+1}$ are on the same level. 
    Without loss of generality, let $j$ be minimal among all these $j$ and say that it is on the $\tilde{\ell}$th level. 
    The intrinsic dimension of this new induced subpolyhedron is at least the intrinsic dimension of $B_{\tilde{\ell}}$ plus the intrinsic dimension of the induced subpolyhedron at the $\mathcal{L}(\lambda'+1)$th level.
    As the (relative) order of the pure power variables remains the same, we preserve property~\ref{step 3a}.

    For property~\ref{step 4c}, it remains to show termination. 
    We do this level-by-level. 
    First, in both cases, when the first level is altered, the intrinsic dimension increases (option 1), or the intrinsic dimension remains unchanged while the ambient dimension decreases (option 2).
    As the intrinsic dimension is bounded by $t$, option 1 occurs finitely often. Meanwhile, as the ambient dimension is between $0$ and $t$, option 2 occurs at most $t$ times in a row. 
    Hence, this process alters the first level only finitely often. 
    By the same argument, the second level is altered finitely often, etc. As the number of levels is bounded by $t$, this process terminates. 

    In a similar manner, we deal with property~\ref{step 3d}. 
    Let $\sum_{k=1}^tc_{i,k}\mathbf{x}_k\equiv c_i \pmod{d_i z_{k_i}^{n_{k_i}}}$ be a modular constraint such that $\mathbf{x}_k$ violates this property with $d_i = z_{k_i}^d$.
    Say $\mathbf{x}_k$ has exponential part $z_j^{n_j}$.
    Then $z_j = z_{k_i}$.
    If $z_j^{n_j} > z_j^{n_{k_i}}$ in the ordering of the power variables, then split into the cases where $z_j^{n_j} = z_jz_j^{n_{k_i}},\dots,z_j^{d-1}z_j^{n_{k_i}}$ (for which Remark~\ref{rem:reduce when equality} allow us to reduce to strictly fewer variables) and $z_j^{n_j} > z_j^{d-1}z_j^{n_{k_i}}$. 
    In the latter case, we have that $\mathbf{x}_k \equiv 0 \pmod{z_j^{d-1}z_j^{n_{k_i}}}$ and so setting $c_{i,k} = 0$ gives an equivalent modular constraint.
    Note that this extra inequality does not change the levels as $n_j$ could not be bounded by $n_{k_i}$ by property~\ref{step 4a}.
    
    When $z_j^{n_j} = z_j^{n_{k_i}}$, we have that $\mathbf{x}_k = x_{i,j}z_j^{n_j}$ is a linear power variable, and we consider all the options for $x_{i,j}$ modulo $z_j^d$ as separate cases. Thus, we write $\mathbf{x}_k = x_{i,j}z_j^{n_j} = (z_j^dx_{i,j}')z_j^{n_j}+rz_j^{n_j}$ for some $0 \le r < z_j^d$ and $x_{i,j}' \in \intg_{\ge 1}$.
    Write $\mathbf{x}_k' = x_{i,j}'z_j^{n_j}$ and eliminate $\mathbf{x}_k$ from the system. 
    This does not introduce a new modular constraint or inequality. 
    Then $c_{i,k}z_j^d \mathbf{x}_k' \equiv 0 \pmod{z_{k_i}^{n_{k_i}+d}}$ by construction, so we can again set $c_{i,k} = 0$.
    Finally, we introduce $\mathbf{x}'_k$ in the ordering of the remaining $t-1$ variables. As $\mathbf{x}'_k < \mathbf{x}_k$, we can only insert it later in the ordering, and as $\mathbf{x}_k/\mathbf{x}'_k$ is bounded, if we introduce it at a later level, the intrinsic dimension of the $\mathcal{L}(k)$th level strictly increases.
    As we only alter linear power variables, property~\ref{step 3a} is preserved.
    
    We have to argue that iterating these algorithms back-to-back ultimately gives both properties.
    We assume that property~\ref{step 4c} holds and then run our algorithm for property~\ref{step 3d}. As there are only finitely many modular constraints with variable modulus, this process will terminate.
    If one of the choices of reinserting $\mathbf{x}'_k$ was on a lower level, the $\mathcal{L}(k)$th level has a strictly larger intrinsic dimension. 
    By the same argument as before, this occurs finitely often for the first level, which will thus be fixed eventually.
    As the number of levels is bounded by $t$, all levels will eventually be fixed by this argument.
    Thus, this process terminates.
\end{proof}

\section{Proofs from Section~\ref{sec:levels and bottom family}}\label{app:levels and bottom family}

\begin{proof}[Proof of Lemma~\ref{lem:Bound for thinness}]
    Applying Farkas' Lemma twice on $(\mathbf{c}_1,\dots,\mathbf{c}_e) \cdot \mathbf{y} = 0$ and $A_{1:e} \mathbf{y} \ge \mathbf{0}$, we can compute rational vectors $\mathbf{w}_-, \mathbf{w}_+ \ge \mathbf{0}$ such that $\mathbf{c}^\top = \mathbf{w}_-^\top A_{1:e} = -\mathbf{w}_+^\top A_{1:e}$.
    Hence, rearranging $A\mathbf{x} > \mathbf{0}$ and using these new vectors $\mathbf{w}_-$ and $\mathbf{w}_+$, we obtain
    \begin{equation*}
        \mathbf{w}_- \cdot (\mathbf{v}-A_{e+1:t}(\mathbf{x}_{e+1},\dots,\mathbf{x}_t)) < \sum_{i=1}^e\mathbf{c}_i\mathbf{x}_i < \mathbf{w}_+ \cdot (A_{e+1:t}(\mathbf{x}_{e+1},\dots,\mathbf{x}_t)-\mathbf{v})
    \end{equation*}
    Now using that $\mathbf{x}_{e+1} \ge 1$ (as $\mathbf{x}_{e+1} \in \intg_{>0}$) and the assumption on $D$, the lemma follows.
\end{proof}

\section{A polyhedron without further constraints}\label{app:no further constraints}

\begin{lemma}\label{lem:reduce to v = 0}
    It is sufficient to prove Theorem~\ref{thm:polyhedron no modular constraints} for $\mathbf{v} = \mathbf{0}$.
\end{lemma}
\begin{proof}
    Assume there is a solution $\mathbf{x}$ in power variables of the correct shape that satisfies $A\mathbf{x} > \mathbf{0}$.
    Split $\mathbf{x}$ and $A$ into two vectors $\mathbf{y}_\alpha$ and $\mathbf{y}_\beta$ and matrices $A_\alpha$ and $A_\beta$, respectively, such that for $z\in\{\alpha,\beta\}$, $\mathbf{y}_z$ contains exactly all power variables in $\mathbf{x}$ with exponential base $z$ and $A_z$ is the restriction of $A$ to the coordinates with exponential base~$z$.
    Thus, $A\mathbf{x} > \mathbf{0}$ if and only if $A_\alpha\mathbf{y}_\alpha+A_\beta\mathbf{y}_\beta >\mathbf{0}$, which we assumed to hold for our original solution.
    Then there is an $\varepsilon > 0$ such that $A_\alpha\mathbf{y}_\alpha+bA_\beta\mathbf{y}_\beta >\mathbf{0}$ whenever $|1-b| < \varepsilon$. 
    By Kronecker's density theorem (Theorem~\ref{thm:kronecker dense}), we find $\beta^n/\alpha^m$ such that $|1 - \beta^n/\alpha^m| < \varepsilon/2$ and $\alpha^m\varepsilon/2 > \max\mathbf{v}$. 
    Then, $\alpha^m \mathbf{y}_\alpha$ and $\beta^n \mathbf{y}_\beta$ consists of power variables of the correct form and
    \begin{equation*}
        \alpha^mA_\alpha\mathbf{y}_\alpha+\beta^n A_\beta\mathbf{y}_\beta > (\alpha^m\varepsilon/2) \mathbf{1} > \mathbf{v}.
    \end{equation*}
    Translating back to $\mathbf{x}$ and $A$, the reduction follows.
\end{proof}

\begin{lemma}\label{lem:lem for first kronecker}
    Let $\varepsilon \in \rat_{>0}$, $e_1 < e_2$ and $0 < c_1 < c_2$ in $\rat$.
    Then there is an $X$ such that for all $x > X$, $a \in \rat$, and functions $f : (e_1 x, e_2 x) \cap \intg \to \rel$ satisfying
    \begin{equation*}
        \frac{c_1}{x} \le f(n) - 2f(n+1) + f(n+2) \le \frac{c_2}{x}\,,
    \end{equation*}
    there is an integer $e_1 x < n < e_2 x$ satisfying $\{f(n) - a\} < \varepsilon$.
\end{lemma}
\begin{proof}
    Let $d = \lceil 2/\varepsilon\rceil$ and $g(n) = f(n) - f(n+1)$. As $g(n+1)-g(n) = -(f(n)-2f(n+1)+f(n+2)) < 0$, $g$ is decreasing.
    Fix an integer $k > \frac{2}{c_1(e_2 - e_1)}$ and put $G(n) = f(n) - f(n+k) = \sum_{i=0}^{k-1} g(n+i)$.
    Then $G(n) - G(n+1) = \sum_{i=0}^{k-1}\big(f(n+i) - 2f(n+i+1) + f(n+i+2)\big) \in [kc_1/x,\, kc_2/x]$ and so $G(n) - G(n+k) \le k^2 c_2/x < \varepsilon/(4(d+1))$ for large enough $x$.
    When $x$ is large enough, there is an $e_1x < n < n+(d+1)k < e_2x$ such that $\{G(n)\} \in (\frac{3}{4}\varepsilon,\varepsilon)$.

    The gaps $G(n), G(n+k), G(n+2k), \dots,G(n+dk)$ decrease in steps smaller than $\varepsilon/(4(d+1))$ and thus $\{G(n+jk)\} \in (\varepsilon/2, \varepsilon)$ for $j = 0,\dots,d$.
    Each step $\{f(n+jk)\} \mapsto \{f(n+(j+1)k)\}$ moves down by an amount in $(\varepsilon/2, \varepsilon)$ modulo $1$; as these sum to at least $1$ and each is less than $\varepsilon$, the points cannot skip $[a, a+\varepsilon)$, and so $\{f(n+jk) - a\} < \varepsilon$ for some $0 \le j \le d$.
\end{proof}

Let $\|x\|$ denote the distance from $x$ to the nearest integer.
\begin{lemma}\label{lem:kronecker consequence 1}
    Let $\alpha, \beta \in\intg_{\ge 2}$ be multiplicatively independent, $a_3, a_4,b_3, b_4, b_5 \in \rat$, and $a_1,a_2,\varepsilon \in \rat_{>0}$ such that $b_3, b_5 \ne 0$. 
    Then there are $x_3, x_4 \in \intg$ and $\gamma \in \rel$ and arbitrarily large $m, n \in \nat$ such that 
    \begin{enumerate}
        \item $|\gamma - a_1| < \varepsilon$; 
        \item $|\frac{\alpha^m}{\beta^n} - a_2| < \varepsilon$;
        \item $|b_3x_3 - \gamma \alpha^m - a_3| < \varepsilon$; and
        \item $|\beta^n + b_4x_3 - \gamma b_5x_4 - a_4| < \varepsilon$.
    \end{enumerate}
\end{lemma}
\begin{proof}
    If the result holds for some $\varepsilon' < \varepsilon$, it holds for $\varepsilon$; so we may shrink $\varepsilon$ and assume $\varepsilon < a_1/2$.
    Hence, every $\gamma$ with $|\gamma - a_1| < \varepsilon$ satisfies $\gamma > a_1/2 > 0$. 
    Replacing $x_3, x_4$ by $-x_3, -x_4$ if necessary, we may assume $b_3, b_5 > 0$.

    For an integer $x_3 = j$, item 3 holds exactly for $\gamma$ within $\varepsilon/\alpha^m$ of $\gamma_j := (b_3 j - a_3)/\alpha^m$, and in particular for $\gamma = \gamma_j$. Item 1 further demands $|\gamma_j - a_1| < \varepsilon$, that is,
    \begin{equation}\label{eq : kronecker 1}
        \frac{(a_1 - \varepsilon)\alpha^m + a_3}{b_3} < j < \frac{(a_1 + \varepsilon)\alpha^m + a_3}{b_3}\,. 
    \end{equation}
    Fix such a $j$ and set $\gamma = \gamma_j$, so items 1 and 3 hold. As $\gamma_j b_5 < 2 a_1 b_5$, item 4 holds for an integer $x_4$ when
    \begin{equation*}
        \left\| \frac{\beta^n + b_4 j - a_4}{\gamma_j b_5} \right\| < \frac{\varepsilon}{2 a_1 b_5}\,.
    \end{equation*}
    Writing the left-hand quantity as $y/(j + b)-a$ with
    \begin{equation*}
        b = -\frac{a_3}{b_3}, \qquad a = -\frac{b_4 \alpha^m}{b_3 b_5}, \qquad y = \frac{\alpha^m}{b_3 b_5}\left(\beta^n - a_4 + \frac{b_4 a_3}{b_3}\right),
    \end{equation*}
    this holds whenever $\{y/(j + b) -a\} < \varepsilon/(2 a_1 b_5)$.

    By Theorem~\ref{thm:kronecker dense}, there are arbitrarily large $m$ and $n$ that satisfy item 2. For these, $\beta^n = (1 + o(1))\alpha^m/a_2$, so $y = (1 + o(1))\alpha^{2m}/(a_2 b_3 b_5)$ lies between $c_1 \alpha^{2m}$ and $c_2 \alpha^{2m}$ for suitable $c_1 < c_2$, while \eqref{eq : kronecker 1} contains an interval $e_1 \alpha^m < j < e_2 \alpha^m$ with $e_1 < e_2$. Applying Lemma~\ref{lem:lem for first kronecker} with $x = \alpha^m$, these $a$ and $b$, and $\varepsilon$ replaced by $\varepsilon/(2 a_1 b_5)$ yields such a $j$. Taking $x_3 = j$, $\gamma = \gamma_j$, and $x_4$ the nearest integer to $(\beta^n + b_4 j - a_4)/(\gamma_j b_5)$ gives the lemma.
\end{proof}

\begin{lemma}\label{lem:kronecker consequence 2}
    Let $\alpha, \beta \in\intg_{\ge 2}$ be multiplicatively independent, $a_3, a_4\in \rat$, $b_3, b_4,b_5 \in \rat_{\ne 0}$ and $a_1, a_2,\varepsilon \in \rat_{>0}$. 
    Then there are $m, n \in \nat$, $x_3, x_4 \in \intg$ and $\gamma \in \rel$ such that 
    \begin{enumerate}
        \item $|\gamma-a_1|< \varepsilon$; 
        \item $|\frac{\alpha^m}{\beta^n} - a_2| < \varepsilon$;
        \item $|x_3 - \gamma b_3 x_4 - a_3| < \varepsilon$; and
        \item $|\alpha^m + b_4x_3 - b_5\gamma \beta^n - a_4| < \varepsilon$.
    \end{enumerate}
\end{lemma}

\begin{proof}
    As before we may shrink $\varepsilon$, so assume $\varepsilon < a_1/2$; then $|\gamma - a_1| < \varepsilon$ forces $2a_1 > \gamma > a_1/2 > 0$.
    Replacing $x_3$ by $-x_3$ if necessary, assume $b_3 > 0$.

    By Theorem~\ref{thm:kronecker dense}, fix arbitrarily large $m, n$ satisfying item 2.
    For an integer $x_3$, item 4 vanishes at $\gamma = \gamma_{x_3} := \frac{\alpha^m + b_4 x_3 - a_4}{b_5\beta^n}$ and holds within $\varepsilon$ in a neighbourhood of it.
    Since $\gamma_{x_3+1} - \gamma_{x_3} = O(\alpha^{-m})$, the condition $\gamma_{x_3} \in (a_1 - \varepsilon, a_1 + \varepsilon)$ of item 1 cuts out an interval $(e_1\alpha^m, e_2\alpha^m)$ of length $\Theta(\alpha^m)$ for choices of $x_3$ for some fixed $e_1 < e_2$.

    For item 3, observe that $2a_1b_3 > \gamma_{x_3}b_3 > 0$ and write
    \begin{equation*}
        P(x_3) = \frac{x_3 - a_3}{\gamma_{x_3}b_3} = \frac{b_5 \beta^n(x_3 - a_3)}{b_3 (\alpha^m + b_4 x_3 - a_4)}\,.
    \end{equation*}
    Item 3 holds whenever $\|P(x_3)\| < \varepsilon/(2a_1 b_3)$ does for a suitable integer $x_4$.
    One of $P(\pm x_3)$ is a rational function with positive second difference $O(\alpha^{-m})$, which is between $c_1/\alpha^m$ and $c_2/\alpha^m$ for fixed $0 < c_1 < c_2$ for all large $m$.
    Applying Lemma~\ref{lem:lem for first kronecker} with $x = \alpha^m$ and $\varepsilon$ replaced by $\varepsilon/(2a_1 b_3)$ yields an integer $x_3$ in this interval with $\|P(\pm x_3)\| < \varepsilon/(2a_1 b_3)$.
    Taking $\gamma = \gamma_{x_3}$ and $x_4$ the nearest integer to $P(\pm x_3)$ gives items 1--4.
\end{proof}

\begin{lemma}\label{lem:nice interval bounded return}
    Let $a \in \rat_{\ne 0}$, $c \in \rat$, $\varepsilon > 0$, and $\mathcal{I}$ an open, non-empty interval. 
    Then there is a number $\ell$ and open, non-empty interval $\mathcal{J} \subseteq \mathcal{I}$ such that for all $\gamma \in \mathcal{J}$ and $M \in \intg$, at least one $x \in \{M,M+1,\dots,M+\ell\}$ and $y \in \intg$ satisfies $|x+a\gamma y+c| < \varepsilon$.
\end{lemma}
\begin{proof}
    Let $p/(qa) \in \rat \cap \mathcal{I}$ be in minimum terms such that $q > 1/\varepsilon$ and assume $\mathcal{J}$ is close around $p/(qa)$, say a distance of less than $\varepsilon/(2|a|q)$.
    Then the map $n \mapsto \frac{pn}{q}$ between $\{0,\dots,q-1\}$ and $\{0,1/q,\dots,(q-1)/q\}$ is one-to-one, giving that $\|\frac{p}{q}y +c\| < \varepsilon/2$ for some $0 \le y < q$. 
    For $0 \le i < q$ and $N \in \intg$, 
    $|(a\gamma N - a \gamma(N+i)) - (\frac{p}{q}N - \frac{p}{q}(N+i))| < \varepsilon/2$.
    Hence, for every $N \in \intg$, we have that $\|a\gamma (N+i) +c\| < \varepsilon$ for some $0 \le i < q$. 
    For each $i$, the corresponding $x = x_i$ equals $-a\gamma(N+i)-c$ rounded to the nearest integer. 
    While varying $N$, these witnesses $x_i$ cannot be farther than $q|a\gamma|+1$ apart, allowing one to take an $\ell \in \nat$ that satisfies $\ell \ge q\sup_{\gamma\in \mathcal{J}}|a\gamma|+1$. 
\end{proof}

\begin{lemma}\label{lem:induction step}
    Let $\alpha,\beta \in \intg_{\ge 2}$ be multiplicatively independent and 
    \begin{equation*}
        A(\mathbf{x}_1, \dots, \mathbf{x}_K, \mathbf{x}_{K+1}, \mathbf{y}_1, \dots, \mathbf{y}_{L}, \mathbf{y}_{L+1}, x', y') > \mathbf{0}
    \end{equation*} be a non-empty polyhedron whose elements satisfy $\mathbf{x}_i,\mathbf{y}_j, x',y' > 0$, and the ratios $\mathbf{x}_i/x'$, $\mathbf{y}_i/x'$, $\mathbf{x}_i/y'$, and $\mathbf{y}_i/y'$ can be arbitrarily large.
    Let $\mathbf{xy}$ be short-hand for $\mathbf{x}_1, \dots, \mathbf{x}_{K+1}, \mathbf{y}_1,\dots,\mathbf{y}_{L+1}$.
    Further, assume the orthogonal complement $U$ of $A_{1:K+L+2}(\mathbf{xy}) \ge \mathbf{0}$ contains no non-zero $(\mathbf{u}, \mathbf{v}) \in \rat^{K+1} \times \rat^{L+1}$ such that $\mathbf{u} = \mathbf{0}$ or $\mathbf{v} = \mathbf{0}$.
    Then there are $x_1,\dots,x_K,y_1,\dots,y_L \in \intg_{\ge 1}$, $k,\ell,m,n \in \nat$ such that $A(\mathbf{xy},x',y') > \mathbf{0}$ when setting $\mathbf{x}_i = x_i\alpha^k$ for $1 \le i \le K$, $\mathbf{y}_i = y_i\beta^\ell$ for $1 \le i \le L$, $x' = \alpha^k$, $y' = \beta^\ell$, $\mathbf{x}_{K+1} = \alpha^m$, and $\mathbf{y}_{L+1} = \beta^n$.
\end{lemma}
\begin{proof}
    Let $(\mathbf{u}_1,\mathbf{v}_1),\dots,(\mathbf{u}_p,\mathbf{v}_p)$ be a basis of $U$.
    By assumption, $p \le \min(K+1, L+1)$.

    We claim that we can assume that $(\mathbf{u}_i)_{K+1}$ is non-zero for at least one $1 \le i \le p$.
    If not, expand the ambient and intrinsic dimension of the polyhedron by 1 by introducing a new variable $\tilde{y}$ and adding the inequalities $\mathbf{x}_{K+1} - \tilde{y} < x'$ and $-\mathbf{x}_{K+1} + \tilde{y} < x'$.
    Then the orthogonal complement of this new polyhedron has increased by one and contains the element belonging to $\mathbf{x}_{K+1} - \tilde{y} = 0$. 
    Similarly, $(\mathbf{v}_i)_{L+1}$ is non-zero for at least one $1 \le i \le p$.

    Applying Gaussian elimination on this basis, we reduce to the case where $(\mathbf{u}_i)_{K+1} \ne 0$ if and only if $i = p$.
    Similarly, $(\mathbf{v}_i)_{L+1} \ne 0$ if and only if $i = \tilde{p}$, where $\tilde{p} \in \{p, p-1\}$.

    With an appropriate change of basis of $\mathbf{x}_1,\dots,\mathbf{x}_K$ that preserves integerness of entries, we get that $(\mathbf{u}_i)_j \ne 0$ if and only if $i = j$ for $1 \le i \le p-1$ and $1 \le j \le K$.
    For $p \le j \le K$, $(\mathbf{u}_i)_j = 0$ for all $i$.
    Project out the $j$th coordinate such that the premise of the lemma remains satisfied.
    By the same argument as Lemma~\ref{lem:reduce to v = 0}, $A(\mathbf{xy},x', y')$ gets large enough that we can take $\mathbf{x}_j$ to be an integer.
    Thus, we reduced to the case $K = p-1$.
    Similarly, we treat the vectors $\mathbf{v}_i$ with $i \ne \tilde{p}$. 
    Let $p'$ denote the unique number in $\{p-1,p\}$ not equal to $\tilde{p}$.
    Then $L = p-1$, $\mathbf{v}_{\tilde{p}} = (0,\dots,0,1)$, $(\mathbf{v}_i)_j \ne 0$ if and only if $i = j$ for $1 \le i \le p-2$ and $1 \le j \le p - 1$, and $(\mathbf{v}_{p'})_j = 0$ for $1\le j \le p-1$.

    Next, by multiplying with an appropriate constant, we can assume that $(\mathbf{u}_i)_i = 1$ for $1 \le i \le p$, and using Gaussian elimination, we can ensure that $(\mathbf{v}_{\tilde{p}})_i = 0$ for $1\le i < p$.
    For $1 \le i \le p-1$, set $b_i = (\mathbf{u}_p)_i$.
    Lastly, let $a \in \intg$ be an integer multiple of $(\mathbf{v}_1)_1,\dots,(\mathbf{v}_{p-2})_{p-2}$ and $(\mathbf{v}_{p'})_{p-1}$. 
    Then, for $1 \le i \le p-2$, there is a $\mathbf{y}_i' \in \intg$ such that $a\mathbf{y}_i = (\mathbf{v}_i)_i \mathbf{y}_i'$ and similarly, for $\mathbf{y}_{p-1}$ and $(\mathbf{v}_{p'})_{p-1}$.
    Hence, instead of these $(\mathbf{v}_i)_i$, we can simply consider $a$, and we write $a'$ for $(\mathbf{v}_{\tilde{p}})_p$.

    Thus, when $\tilde{p} = p$, we have reduced to
    \begin{equation*}
        \begin{pmatrix}
            \mathbf{u}_1^\top & \mathbf{v}_1^\top \\ 
            \vdots & \vdots \\
            \mathbf{u}_p^\top & \mathbf{v}_p^\top
        \end{pmatrix}
        =
        \begin{pmatrix}
            1 & 0 & \cdots & 0 & 0 & a & 0 & \cdots & 0& 0 \\ 
            0 & 1 & \cdots & 0 &0 & 0 & a & \cdots & 0& 0 \\ 
            \vdots & \vdots & \ddots& \vdots & \vdots & \vdots & \vdots & \ddots & \vdots& \vdots \\ 
            0 & 0 & \cdots & 1 & 0 & 0 & 0 & \cdots & a& 0 \\ 
            b_1 & b_2 & \cdots & b_{p-1} & 1 & 0 & 0 & \cdots & 0& a' \\ 
        \end{pmatrix}
    \end{equation*}
    and, when $\tilde{p} = p-1$, we have reduced to
    \begin{equation}\label{eq:matrix ptilde = p-1}
        \begin{pmatrix}
            \mathbf{u}_1^\top & \mathbf{v}_1^\top \\ 
            \vdots & \vdots \\
            \mathbf{u}_p^\top & \mathbf{v}_p^\top
        \end{pmatrix}
        =
        \begin{pmatrix}
            1 & 0 & \cdots & 0 & 0 & a & 0 & \cdots & 0& 0 \\ 
            0 & 1 & \cdots & 0 &0 & 0 & a & \cdots & 0& 0 \\ 
            \vdots & \vdots & \ddots& \vdots & \vdots & \vdots & \vdots & \ddots & \vdots& \vdots \\ 
            0 & 0 & \cdots & 1 & 0 & 0 & 0 & \cdots & 0& a' \\ 
            b_1 & b_2 & \cdots & b_{p-1} & 1 & 0 & 0 & \cdots & a& 0 \\ 
        \end{pmatrix}\,.
    \end{equation}
    As the polyhedron $A(\mathbf{xy},x',y') > \mathbf{0}$ is non-empty, it contains an element $(\widetilde{\mathbf{xy}},\tilde{x}',\tilde{y}') \in\rat^{2p} \times \rat_{>0}^2$.
    By the strict inequality, we can change $\tilde{x}'$ and $\tilde{y}'$ slightly. 
    That is, when $g \coloneqq \tilde{x}'/\tilde{y}'$, $\varepsilon > 0$, and $x',y'>0$ satisfy $x'/y' \in (g - \varepsilon, g+\varepsilon)$ we have that $A(\widetilde{\mathbf{xy}},x',y') > \mathbf{0}$.
    As any $0<\varepsilon'< \varepsilon$ also has this property, we can assume that $\varepsilon < g/2$.
    Thus, when choosing $x' > 0$, we have that $y'>0$ and that $x'/y'$ is bounded when $x'/y' \in (g - \varepsilon, g+\varepsilon)$. 

    For $1 \le j \le p$, let $c_j = \mathbf{u}_j \cdot \tilde{\mathbf{x}} + \mathbf{v}_j \cdot\tilde{\mathbf{y}}$.
    Then, for small enough $\varepsilon > 0$, we have that $x'/y' \in (g - \varepsilon, g+\varepsilon)$ and $|\mathbf{u}_j \cdot \mathbf{x} + \mathbf{v}_j \cdot\mathbf{y}-c_jx'|/x' < \varepsilon$ imply that for any row $A_i$ of $A$ that is in $U$, we have that $A_i(\mathbf{xy},x',y') > 0$ as each such $A_i$ is a linear combination of $(\mathbf{u}_j, \mathbf{v}_j)$ and $A_{i,2p+1}+A_{i,2p+2}y'/x'$ is $O(\varepsilon)$ away from $A_{i,2p+1}+A_{i,2p+2}\tilde{y}'/\tilde{x}'$.

    By assumption, there is a non-zero vector $\widetilde{\mathbf{xy}}$ such that $A_{1:2p}\widetilde{\mathbf{xy}} \ge \mathbf{0}$. 
    In particular, there is a $p$-dimensional open cone for which for every $\widetilde{x}$ within, there is a unique $\widetilde{\mathbf{y}}$ such that $A_{1:2p}\widetilde{\mathbf{xy}} \ge \mathbf{0}$. Hence, we can assume that $(a_1,\dots,a_p) := \widetilde{x} \in \rat_{\ne0}^p$.
    Then, if the projection of $\mathbf{xy}$ on the first $p$ coordinates is in a narrow enough cone around $(a_1,\dots,a_p)$ and its entries are much larger than $x'$ and $y'$, $A_i \cdot (\mathbf{xy},x',y') > \mathbf{0}$ for every row $A_i$ outside of $U$. 
    Thus, there is a number $N$ and a small enough $\varepsilon$ such that $A(\mathbf{xy},x',y') > \mathbf{0}$ whenever
    \begin{enumerate}
        \item $|\mathbf{x}_j/\mathbf{x}_p - a_j/a_p| < \delta$ for $1 \le j \le p-1$ and $\mathbf{x}_p > 0$;
        \item $|\mathbf{x}_p/x'| > N$;
        \item $|\mathbf{u}_j \cdot \mathbf{x} + \mathbf{v}_j \cdot \mathbf{y}-c_jx'|/x'< \varepsilon$  for $1 \le j \le p$; and
        \item $x'/y' \in (g - \varepsilon, g+\varepsilon)$.
    \end{enumerate}
    Moreover, this still holds when tightening $\delta > 0$. Similarly, we can update $g$ and $\varepsilon$ to $\tilde{g}$ and $\tilde{\varepsilon}$ when $(\tilde{g} - \tilde{\varepsilon}, \tilde{g}+\tilde{\varepsilon}) \subseteq (g - \varepsilon, g+\varepsilon)$.
    Then $\tilde{\varepsilon} < \tilde{g}/2$ still holds.
    By updating $g$ and $\varepsilon$, we can also write $y'/x' \in (g - \varepsilon, g+\varepsilon)$ for item 4.

    We claim that when the statement holds for a real $\gamma \in (g - \varepsilon, g+\varepsilon)$ and $\mathbf{x}_p > N$, we can take $\gamma$ of the form $y'/x'$ for $y'\in \beta^\nat$ and $x'\in\alpha^\nat$.
    Indeed, due to the strong inequalities and the denseness of $\{\beta^n/\alpha^m: m,n\in \nat\}$ in $\rel_{\ge0}$ (Theorem~\ref{thm:kronecker dense}), there are $y'/x'$ arbitrarily close to $\gamma$. 
    Now set $\mathbf{x}' = x'\mathbf{x}$, $\mathbf{y}' = y'\mathbf{y}$, and items 1-4 are satisfied.

    Let us unpack these definitions further. Recall that $\mathbf{x}_p, x'\in \alpha^\nat$, $\mathbf{y}_p,y' \in \beta^\nat$, $\mathbf{x}_1,\dots,\mathbf{x}_{p-1}$ are integer multiples of $x'$, and $\mathbf{y}_1,\dots,\mathbf{y}_{p-1}$ are integer multiples of $y'$. 
    If we write $a_j$ for $a_j/a_p$, $\gamma = y'/x'$, and divide $\mathbf{x}$ by $x'$ and $\mathbf{y}$ by $y'$, it is sufficient to find $\gamma \in \rel$, $\mathbf{x}, \mathbf{y} \in \intg^p$, $\mathbf{x}_p\in\alpha^\nat$, and $\mathbf{y}_p\in\beta^\nat$ that satisfy
    \begin{enumerate}
        \item $|\mathbf{x}_j/\mathbf{x}_p - a_j| < \delta$ for $1 \le j \le p-1$;
        \item $\mathbf{x}_p > N$;
        \item $|\mathbf{u}_j \cdot \mathbf{x} + \gamma\mathbf{v}_j \cdot \mathbf{y} - c_j|< \varepsilon$  for $1 \le j \le p$; and
        \item $\gamma \in (g - \varepsilon, g+\varepsilon)$.
    \end{enumerate}
    Here, we might have updated $g$ and $\varepsilon$ and used that $x'> 0$.
    
    We return to the basis of our orthogonal complement. 
    If $b_j = 0$ for some $1 \le j < \tilde{p}$, then the only constraints involving $\mathbf{x}_j$ and $\mathbf{y}_j$ are that $|\mathbf{x}_j/\mathbf{x}_p - a_j| < \delta$ and that $|\mathbf{x}_j + \gamma a\mathbf{y}_j-c_j| < \varepsilon$.
    By Lemma~\ref{lem:nice interval bounded return}, there is a number $\ell$ and an interval $\mathcal{J} = (\tilde{g} - \tilde{\varepsilon}, \tilde{g}+\tilde{\varepsilon})$ such that for every $\gamma \in \mathcal{J}$ and $M \in \nat$, the latter is satisfied by at least one $\mathbf{x}_j \in \{M,M+1,\dots,M+\ell\}$ for some $\mathbf{y}_j$.
    If $\mathbf{x}_1$ is large enough, all of those satisfies $|\mathbf{x}_j/\mathbf{x}_p - a_j| < \delta$ for some $M$.
    Thus, potentially updating $N$, these constraints have no effect, and we can assume $b_j \ne 0$ for $j < \tilde{p}$.

    We reduce to $p \le 2$.
    First, consider $\tilde{p} > 2$. 
    We want to write $\mathbf{x}_2 = r \mathbf{x}_1 + s$ for $r, s \in \rat$ such that when $\mathbf{x}_1$ lies in a certain congruence class, $\mathbf{x}_2$ is also an integer.
    Taking $0 < \delta' < \delta$ small enough, there is certainly a $r \in \rat$ with arbitrarily large denominator such that $|\mathbf{x}_1/\mathbf{x}_p - a_1| < \delta'$ implies that $|r\mathbf{x}_1/\mathbf{x}_p - a_2| < \delta$. For a fixed $s$ and large enough $N$, $|\mathbf{x}_2/\mathbf{x}_p - a_2|=|r\mathbf{x}_1/\mathbf{x}_p+s/\mathbf{x}_p - a_2| < \delta$ holds as well.
    There are integers $\mathbf{y}_1$ and $\mathbf{y}_2$ satisfying  $|\mathbf{x}_1 + \gamma a \mathbf{y}_1 - c_1|,|r\mathbf{x}_1+s + \gamma a \mathbf{y}_2 - c_2| < \varepsilon$ when
    \begin{equation}\label{eq:rewrite thinness 1}
        \left\|\frac{1}{a\gamma}(\mathbf{x}_1 - c_1)\right\|\,,\,\left\|\frac{1}{a\gamma}(r\mathbf{x}_1+s - c_2)\right\| < \varepsilon\,.
    \end{equation}
    This will always hold when 
    \begin{equation}\label{eq:rewrite thinness 2}
        \left\|\frac{1}{a\gamma}(\mathbf{x}_1- c_1)\right\|< \frac{\varepsilon}{2r|a|(g+\varepsilon)}\,,\,\left\|\frac{1}{a\gamma}(rc_1+s -c_2)\right\| < \frac{\varepsilon}{2|a|(g+\varepsilon)}\,.
    \end{equation}
    From this form, it is clear that $s$ exists when $\gamma$ is in a small enough specific interval and $r$ has a large enough denominator, and that $\mathbf{x}_2$ is always an integer when $\mathbf{x}_1$ lies in a certain congruence class. Making the substitution $\mathbf{x}_2 = r\mathbf{x}_1+s$ into the $j = p$-case in item 3 fully eliminates $\mathbf{x}_2$.

    If $p > 2$ and $\tilde{p} \le 2$, we have that $p = 3$ and $\tilde{p} = 2$. 
    We apply the same method but have to swap the roles of $\mathbf{x}$ and $\mathbf{y}$. 
    The difficulty is that we are dealing with $\mathbf{x}_1$, $\mathbf{x}_2$, $\mathbf{y}_1$, and $\mathbf{y}_3$, where the last is a power of $\beta$. 
    First use linear algebra on the last row of the matrix~\eqref{eq:matrix ptilde = p-1} and potentially update $\varepsilon$ in the process to get a matrix
    \begin{equation*}
        \begin{pmatrix}
            1 & 0 & 0 & a & 0 & 0\\
            0 & 1 & 0 & 0 & 0 & a'\\
            0 & 0 & 1 & b_1' & a & b_3'
        \end{pmatrix}\,,
    \end{equation*}
    which corresponds to item 3. 
    Then write $\mathbf{y}_1 = r\mathbf{y}_3+s$ and apply the exact same method to eliminate $\mathbf{y}_1$. 
    When $\varepsilon$ is small enough and $\mathbf{y}_1$ and $\mathbf{y}_3$ are large enough, then item 3 implies that $\mathbf{x}_1$ and $\mathbf{y}_1$ are proportional in $(a\gamma)^{-1}$ and $\mathbf{x}_2$ and $\mathbf{y}_3$ in $(a\gamma)^{-1}$. Hence, if $\mathbf{y}_1 = r\mathbf{y}_3$, then $|\mathbf{x}_1/\mathbf{x}_p - a_1| < \delta'$ implies that $|\mathbf{x}_2/\mathbf{x}_p - a_2| < \delta$ for some $0 < \delta' < \delta$, giving the new item 1 when $r$ is chosen from a specific interval.
    Equivalent to~\eqref{eq:rewrite thinness 1} and~\eqref{eq:rewrite thinness 2}, we get that we need that
    \begin{equation*}
        \|\gamma a (r\mathbf{y}_3+s)-c_1\| \,,\,\|\gamma a \mathbf{y}_3-c_2\| < \varepsilon\,,
    \end{equation*}
    which follows from
    \begin{equation*}
        \| r\mathbf{y}_3c_2+\gamma as-c_1\| < \frac{\varepsilon}{2|a|(g+\varepsilon)} \,,\,\|\gamma a \mathbf{y}_3-c_2\| < \frac{\varepsilon}{2r|a|(g+\varepsilon)}\,.
    \end{equation*}
    Now taking $r$ and $s$ in these intervals, such that for some $c, c'$, $r \beta^{cn+c'}+s$ is an integer for all $n$. 
    Then make the natural substitution $\beta \mapsto \beta^c$ and apply the same linear algebra to return to the original shape~\eqref{eq:matrix ptilde = p-1}.

    Thus, we can assume that $p \le 2$.
    The remaining two cases for $p =2$ are exactly Lemmas~\ref{lem:kronecker consequence 1} and~\ref{lem:kronecker consequence 2}.
    The remaining case, $p = 1$, follows easily: Item 1 is trivially satisfied.
    Item 3 is equivalent to $|\alpha^m+\gamma a' \beta^n - c_1| < \varepsilon$.
    As we started with a satisfiable system for arbitrarily large $\mathbf{x}_1$ and $\mathbf{y}_1$, we must have $a' < 0$. 
    Then taking $-\alpha^m/(a'\beta^n) \in (g - \varepsilon/2, g+\varepsilon/2)$ with $\alpha^m > N$ large enough (which exist by Theorem~\ref{thm:kronecker dense}), we can take $\gamma$ such that $\alpha^m + \gamma a'\beta^n - c_1 = 0$.
\end{proof}

\begin{proof}[Proof of Theorem~\ref{thm:polyhedron no modular constraints}]
    By Lemma~\ref{lem:reduce to v = 0}, we can assume that $\mathbf{v} = \mathbf{0}$.

    We claim we can assume the highest level $L$ contains exactly two power variables, which are pure variables of different exponential bases.
    Indeed, properties~\ref{step 1} and~\ref{step 4b} imply that for any linear power variable with exponential part $z_j^{n_j}$, there is a pure power variable $\mathbf{x}_k = z_j^{n_j}$ on a higher level, which, as $L$ is the highest level, cannot exist. 
    By property~\ref{step 4a}, there is at most one pure power variable of each exponential base at the $L$th level. 
    As levels are non-empty, there is already at least one pure power variable $\mathbf{x}_t$.
    If there is a second, the claim follows.
    Otherwise, introduce a dummy variable $\mathbf{x}_{t+1}$ such that $\mathbf{x}_{t+1} < \mathbf{x}_t < 2\mathbf{x}_{t+1}$ and update $A$ and $\mathbf{v}$ accordingly.
    It is easily verified that all variables $\mathbf{x}_1,\dots,\mathbf{x}_t$ remain on their original levels, and $\mathbf{x}_t$ and $\mathbf{x}_{t+1}$ are on the same level while the hypothesis of the lemma remains satisfied. 
    The claim follows.

    We complete the argument by induction on $L$.
    
    \emph{The base case $L = 1$}. 
    There are exactly two variables at this level (being the highest), namely, pure power variables of both bases. 
    Thus, we need to find $m, n \in \nat$ such that $A(\alpha^m, \beta^n) > \mathbf{0}$.
    As the polyhedron $A\mathbf{x} > \mathbf{0}$ is not thin (property~\ref{step 2}) and has ambient dimension 2, there are $c \in \rat$ and $\varepsilon > 0$ such that $A(1, b) > \mathbf{0}$ whenever $|b-c| < \varepsilon$. 
    By Kronecker's theorem (Theorem~\ref{thm:kronecker dense}), there are $m, n$ such that $\beta^n/\alpha^m \in (c - \varepsilon, c+ \varepsilon)$ and so by homogeneity, $A(\alpha^m, \beta^n) > \mathbf{0}$.
    
    \emph{Inductive step}. 
    Assume we have $L \ge 2$ levels of which the lowest is $\mathcal{L}^{-1}(1) = \{1,\dots,\lambda\}$.
    By possibly introducing dummy variables, we can assume there are exactly two pure power variables of differing exponential base among $\mathbf{x}_1',\dots,\mathbf{x}_{\lambda}'$.
    By the induction hypothesis, there are power variables $\mathbf{x}_{\lambda+1}',\dots,\mathbf{x}_t'$ of the correct shape such that $A'(\mathbf{x}_{\lambda+1}',\dots,\mathbf{x}_t') > \mathbf{0}$, where this polyhedron is the projection onto all but the lowest level. 

    We want to relabel polyhedra and vectors to split them by their exponential bases and levels. 
    For $z \in \{\alpha,\beta\}$, let $\mathbf{y}^{(z)}$ collect the power variables $\mathbf{x}_{\lambda+1}',\dots,\mathbf{x}_t'$ with exponential base $z$ and $A_z'$ be the restriction of $A_{\lambda+1:t}$ to these power variables. 
    Thus, for any $n \in \nat$, $z^n \mathbf{y}^{(z)}$ also consists of power variables of the correct shape and  $A_\alpha'\mathbf{y}^{(\alpha)} + A_\beta'\mathbf{y}^{(\beta)} = A_{\lambda+1:t}(\mathbf{x}_{\lambda+1}',\dots,\mathbf{x}_t')$.
    Let $\ell_z$ denote the exponent of the smallest pure power variable with exponential basis $z$. 
    Thus, each $\mathbf{y}^{(z)}_k$ is a multiple of $z^{\ell_z}$.
    We keep $\ell_\alpha$ and $\ell_\beta$ variable while using our known solution to scale it, i.e., if $\tilde{\mathbf{y}}^{(z)}$ is part of a solution of the last $L-1$ levels, we write $\mathbf{y}^{(z)} = z^{\kappa_z}\tilde{\mathbf{y}}^{(z)}$ for some $\kappa_z \in \nat$. 
    
    We treat the power variables on the first level similarly. 
    For $z \in \{\alpha,\beta\}$, let $\tilde{\mathbf{x}}^{(z)}$ collect the power variables $\mathbf{x}_{1},\dots,\mathbf{x}_\lambda$ with exponential base $z$.
    Let $K_z+1$ be the dimension of $\tilde{\mathbf{x}}^{(z)}$ and assume $\tilde{\mathbf{x}}^{(z)}_{K_{z}+1}$ is the sole pure power variables in $\tilde{\mathbf{x}}^{(z)}$.
    For $1 \le k \le K_z$, we write $\tilde{\mathbf{x}}_k^{(z)} = x_{i,j}z^{\ell_z}$, which we put together as 
    \begin{equation*}
        A(\mathbf{x}_1,\dots,\mathbf{x}_t) = \tilde{A}(x_1\alpha^{\ell_\alpha},\dots,x_{K_\alpha}\alpha^{\ell_\alpha},\alpha^{n_\alpha},y_1\beta^{\ell_\beta},\dots,y_{K_\beta}\beta^{\ell_\beta},\beta^{n_\beta},\alpha^{\ell_\alpha},\beta^{\ell_\beta})\,.
    \end{equation*}
    This form is exactly as in Lemma~\ref{lem:induction step} whose technical condition is property~\ref{step 4c}.
\end{proof}

\section{Proofs from Section~\ref{sec:Finalising the argument}}\label{app:Finalising the argument}

\begin{proof}[Proof of Lemma~\ref{lem:range modulo}]
    Assume $\ell$ has the following property: For $m \ge \ell$, $\beta$ has order $r$ in $(\intg/\alpha^m \intg)^*$ and order $\alpha r$ in $(\intg/\alpha^{m+1} \intg)^*$.
    Then, by induction, $\beta$ having order $r$ in $(\intg/\alpha^\ell \intg)^*$ implies that $\beta$ has order $\alpha^{m-\ell}r$ in $(\intg/\alpha^m \intg)^*$. If $0 \le a < \alpha^\ell$, $m \ge \ell$, $\tilde{a} \equiv a \pmod{\alpha^m}$, and $\beta^n \not\equiv a \pmod{\alpha^\ell}$ for all $n \in \nat$, we have that $\beta^n \not\equiv \widetilde{a} \pmod{\alpha^m}$ for all $n \in \nat$.
    The pigeonhole principle implies that $\ell$ satisfies the lemma.
    After computing $\ell$, take $S = \{\beta^i \bmod{\alpha^\ell} : 0 \le i < r\}$, where $r$ is the order of $\beta$ in $(\intg/\alpha^\ell \intg)^*$.

    It remains to compute such $\ell$.
    First, consider $\alpha = p$ being prime. 
    For $m \ge 2$ and $1 \le c < p$, consider $x = 1+cp^{m-1}$. 
    The binomial theorem gives
    \begin{equation*}
        x^p = (1+cp^{m-1})^p = \sum_{i=0}^p \binom{p}{i}c^ip^{(m-1)i} \equiv 1 + cp^m + \frac{p-1}{2}c^2p^{2m-1} \pmod{p^{m+1}}
    \end{equation*}
    as $(m-1)i \ge m+1$ for $i \ge 3$ and $m \ge 2$.
    Hence, for $m \ge 3$ (due to the case $p=2$), we have that $x^p \equiv 1+cp^m \pmod{p^{m+1}}$.
    Therefore, if $m \ge 3$ and $x$ has order $r$ in $(\intg/p^m\intg)^*$, $x^p$ has order $pr$. 
    Thus, one can take $\ell \ge 3$ such that $\beta^r \equiv 1 \pmod{\alpha^{\ell-1}}$ and $\beta^r \not\equiv 1 \pmod{\alpha^\ell}$.

    When $\alpha$ is a prime power, say $\alpha = p^r$, let $\ell'$ be the constructed number for the prime $p$. 
    Then take $\ell$ such that $\ell \ge \ell'/r$.
    Now for general $\alpha$, one factors $\alpha$ as $p_1^{r_1}\cdots p_b^{r_b}$ and compute $\ell_i$ for each $p_i^{r_i}$. 
    Then let $\ell \coloneqq \max_{1 \le i \le b}\ell_i$, which is correct using the Chinese remainder theorem.
\end{proof}

We strengthen Theorem~\ref{thm:polyhedron no modular constraints} by allowing variables that are not power variables. 
\begin{corollary}\label{cor:polyhedron no modular constraints + lin variables}
    Let $A\mathbf{x}+B\mathbf{y} > \mathbf{v}$ be a non-thin polyhedron where $\mathbf{x}$ consists of power variables and the projection onto the variables in $\mathbf{x}$ satisfies properties~\labelcref{step 1} and~\labelcref{step 4a,step 4b,step 4c}.
    For every $R > 0$, there is a solution in power variables $\mathbf{x}_k$ and integers $\mathbf{y}_k$ such that $\mathbf{x}_k/\mathbf{x}_{k+1} > R$ whenever $\mathbf{x}_k$ and $\mathbf{x}_{k+1}$ are on different levels.
\end{corollary}
\begin{proof}
    Let $\tilde{b} = \max_{i,j} |B_{i,j}|$.
    The projection of the (non-thin) polyhedron $A\mathbf{x}+B\mathbf{y} > \mathbf{v} + s\tilde{b}\mathbf{1}$ onto the coordinates of $\mathbf{x}$ is non-thin and satisfies the hypothesis of Theorem~\ref{thm:polyhedron no modular constraints}. Let $\mathbf{x}_1,\dots,\mathbf{x}_t$ be power variables in this projection.
    Then there are rational numbers $\mathbf{y}_1',\dots,\mathbf{y}_s'$ such that $A\mathbf{x}+B\mathbf{y}' > \mathbf{v} + s\tilde{b}\mathbf{1}$.
    We apply induction on $j$ such that $\mathbf{y}_1,\dots,\mathbf{y}_j$ are integers and $A\mathbf{x}+B(\mathbf{y}_1,\dots,\mathbf{y}_j,\mathbf{y}_{j+1}',\dots,\mathbf{y}_s') > \mathbf{v} + (s-j)\tilde{b}\mathbf{1}$.
    For $j = 0$, this trivially holds.
    If the induction holds for $j$, then for $\mathbf{y}_{j+1}$ take the single integer in $[\mathbf{y}_{j+1}',\mathbf{y}_{j+1}'+1) \cap\intg$.
    Then the induction step follows from the definition of $\tilde{b}$, and the corollary follows when taking $j = s$.
\end{proof}

\begin{proof}[Proof of Theorem~\ref{thm:main theorem}]
    Compute the bottom family $\mathbf{x}_1,\dots,\mathbf{x}_b$ and a basis $\mathbf{u}_1,\dots,\mathbf{u}_p$ of the orthogonal complement $U$ of the (recession) polyhedron $A_{1:b}(\mathbf{x}_1,\dots,\mathbf{x}_b) \ge \mathbf{0}$. 
    We can assume that for each $1 \le i \le p$ there is a $1 \le u_i \le b$ such that $\mathbf{x}_{u_i}$ does not have small coefficients and $(\mathbf{u}_j)_{u_i} \ne 0$ if and only if $i = j$. 
    We achieve this using induction on $I = 0,\dots,p$, where we satisfy the property for $1 \le i \le I$ and eliminate $\mathbf{x}_{u_1},\dots,\mathbf{x}_{u_I}$ from all other basis elements. Then the properties of the bottom family say that a power variable $\mathbf{x}_{u_{I+1}}$ without small coefficients exists.
    Moreover, at each step, we choose the power variable $\mathbf{x}_{u_i}$ to have minimal exponential part.
    Multiplication with an appropriate constant ensures all $\mathbf{x}_i$ are integer-valued.

    For $1 \le i \le p$, set $x_i = \mathbf{u}_i \cdot (\mathbf{x}_1,\dots,\mathbf{x}_b)$ symbolically, which we treat as a usual integer variable and also quantify in the quantifier elimination. 
    Use these equations to eliminate $\mathbf{x}_{u_i}$ from the polyhedron $A$ and all modular equations such that $\mathbf{x}_{u_i}$ only appears in these equations, which we can turn into a modular constraint with variable modulus that by construction satisfies properties~\labelcref{step 3a,step 3c} already. 
    We ensure that property~\ref{step 3b} holds by applying the proof of Lemma~\ref{lem:d_i power of z_(k_i) 3b} \emph{mutatis mutandis}.
    Moreover, when $\mathbf{x}_k,\mathbf{x}_{k+1}$ are not on the same level, then $\mathbf{x}_k/\mathbf{x}_{k+1}$ is bounded from below.

    Rename the $t$ variables and $b$ such that again the first $b$ variables are power variables from the bottom family, and that $\mathbf{x}_{b+1},\dots,\mathbf{x}_t$ are power variables and the new integer-valued variables $x_i$.
    We can assume that $b \ge 1$ as $b = 0$ would allow us to reduce to fewer variables, as we have at least one integer-valued variable and can apply Remark~\ref{rem:reduce when equality}.

    Let us deal with the modular constraint with constant modulus.
    For such a modular constraint $\sum_{k=1}^t c_{i,k}'\mathbf{x}_k \equiv c_i' \pmod{d_i'}$, compute all $(s_1,\dots,s_{b+1}) \in \{0,\dots,d_i'-1\}^{b+1}$ such that $\sum_{k=1}^b c_{i,k}'s_k + s_{b+1} \equiv c_i' \pmod{d_i'}$.
    Then split into the cases for these tuples $(s_1,\dots,s_{b+1})$ and replace $\sum_{k=1}^t c_{i,k}'\mathbf{x}_k \equiv c_i' \pmod{d_i'}$ by $\mathbf{x}_k \equiv s_k \pmod{d_i'}$ and $\sum_{k=b+1}^t c_{i,k}'\mathbf{x}_k \equiv s_{b+1} \pmod{d_i'}$.

    The modular constraints with variable modulus are more complicated.
    To start, let $\sum_{k=1}^t c_{i,k}\mathbf{x}_k \equiv c_i \pmod{d_iz_{k_i}^{n_{k_i}}}$ be a modular constraint with variable modulus where all $c_{i,k} \in \intg$ and $d_i$ is a power of $z_{k_i}$.
    
    Let $d$ be such that each $d_i$ divides $z_{k_i}^d$. 
    Then, can reduce to the case where $d_i = z_{k_i}^d$ by considering the cases $ez_{k_i}^{n_{k_i}}+  \sum_{k=1}^t c_{i,k}\mathbf{x}_k \equiv c_i \pmod{d_iz_{k_i}^{n_{k_i}}}$ for all $0 \le e < d$ such that $d_i \mid e$. 
    Next, recall that for pure power variables $\mathbf{x}_k= z_j^{n_j}$ and $\mathbf{x}_{k'} = z_{j'}^{n_{j'}}$ with $k < k'$ we have that $\mathbf{x}_{k'}/\mathbf{x}_k$ is bounded from above. 
    Hence, when $z_j = z_{j'}$, we can split into the case where $n_j > n_{j'}$ and  the finitely many case where $n_j = n_{j'} + c$ for a constant $c \in \nat$.
    Thus, we can assume that pure power variables of a specific exponential base are ordered strictly.
    
    Similarly, we can reduce to the case where $c_{i,k} = 0$ for power variables $\mathbf{x}_k$ in the bottom family with the exponential part $z_{k_i}^{n_j}$ (thus the same exponential base).
    If $n_j > n_{k_i}$, we again split into the cases where $n_j - n_{k_i} = 0,\dots,d-1$ and else $\mathbf{x}_k \equiv 0 \pmod{z_{k_i}^{n_{k_i}+d}}$ and so we can equivalently set $c_{i,k} = 0$.
    When $n_j = n_{k_i}$, $\mathbf{x}_k$ is a linear power variable (properties~\labelcref{step 3a,step 3c}), and so we split into the cases where $x_{i,j} z_{k_i}^{n_j} = (d\tilde{x}_{i,j}+e) z_{k_i}^{n_j}$ for $0 \le e < d$ to obtain that $\mathbf{x}_k$ is always divisible by $z_{k_i}^{n_{k_i}+d}$, allowing us to set $c_{i,k} = 0$. 
    When $n_j < n_{k_i}$, using properties~\labelcref{step 3a,step 3c}, we have that $\mathbf{x}_k$ is outside the bottom family.
    
    Finally, observe that when $a_ix_i \equiv 0\pmod{e_i}$ for $i = 1,2$ and $e_1 \mid e_2$, then $a_1x_1+na_2x_2\equiv 0\pmod{e_1}$ for any $n \in \intg$, i.e., we can use a restricted form of Gaussian elimination on these constraints. 
    First, as all moduli have the form $z_{k_i}^{d+n_{k_i}}$, pure power variables of a specific exponential base are strictly ordered, we can rename the modular constraints such that for $i < i'$ with $z_{k_i} = z_{k_{i'}}$ that $n_{k_i} \le n_{k_{i'}}$.
    As $c_{i,k} \ne 0$ when $\mathbf{x}_k$ has opposite exponential base of $z_{k_i}$, we can ensure by this limited Gaussian elimination that for each equation either $c_{i,k} = 0$ for all $1 \le k \le b$ or there is a $1 \le k\le b$ such that $c_{i,k} \ne 0$ and $c_{i',k} = 0$ for all $1 \le i' < i$.
    This process might increase $d$, as in the example above, $a_2$ might have more factors in common with $c_1$ than $a_1$ does (we do not make this more specific here). 

    To finally deal with the modular constraints with variable modulus, compute a number $\ell_i$ such that, for every $1 \le k \le b$ with $c_{i,k} \ne 0$ (and thus $\mathbf{x}_k$ has exponential base different than $z_{k_i}$), we have that $z_{k_i}^{\ell_i}/c_{i,k}$ contains enough factors of $z_{k_i}$ to apply Lemma~\ref{lem:range modulo}. 
    Then, as in the constant modulo case, split into the cases for $(s_1,\dots,s_{b+1}) \in \{0,\dots,z_{k_i}^{\ell_i}-1\}^{b+1}$ where $\sum_{k=1}^bc_{i,k}s_k + s_{b+1}\equiv c_i\pmod{z_{k_i}^{\ell_i}}$ and for $1 \le k \le b$, $s_k \in S$ whenever $\mathbf{x}_k$ is a pure power variable, and add the modular constraints with $\mathbf{x}_k \equiv s_k \pmod{z_{k_i}^{\ell_i}}$ for $1 \le k \le b$ and $\sum_{k=b+1}^t c_{i,k}\mathbf{x}_k \equiv s_{b+1} \pmod{z_{k_i}^{\ell_i}}$.

    Combine all these conditions on $\mathbf{x}_k$ being something modulo something into one constraint $\mathbf{x}_k \equiv s_k \pmod{e_k}$ using the Chinese remainder theorem (and return false if such an $s_k$ does not exist). 
    For a pure power variable $\mathbf{x}_k = z_j^{n_j}$ and modular constraint $\mathbf{x}_k \equiv s_k \pmod{e_k}$, we have that $(z_j^{n} \bmod{e_k})_{n=0}^\infty$ is ultimately periodic.
    As the next value only depends on the previous, we either have that there is no solution $n_j$ (leading to a direct contradiction), exactly one $n_j$ (leading to a reduction to a system with fewer power variables), or we can compute $a_j, b_j, c_j$ such that $z_j^{n_j} \equiv s_k \pmod{e_k}$ holds exactly when $n_j \equiv a_j \pmod{a_j'}$ and $n_j > c_j$.
    Similarly, when $\mathbf{x}_k = x_{i,j}z_j^{n_j}$ is a linear power variable, we can compute $b_k, c_k$ such that $z_j^{n_j} \equiv z_j^{n_j+b_k} \pmod{e_k}$ for all $n_j \ge c_j$. 
    Splitting into the cases where $n_j < c_j$ (which are systems with strictly fewer variables), and $n_j \equiv b_k \pmod{b_k'}$ and $n_j \ge c_j$ for $a_k = 0,\dots,b_k-1$.
    We reduce to the case where we have infinitely many options for each $\mathbf{x}_{k}$.

    Now we make our quantifier elimination. 
    Find a solution $\tilde{\mathbf{x}}$ to our problem with $\tilde{\mathbf{x}}_1,\dots,\tilde{\mathbf{x}}_b$ integer variables satisfying $\tilde{\mathbf{x}}_k \equiv s_k \pmod{e_k}$ and $\tilde{\mathbf{x}}_{b+1},\dots,\tilde{\mathbf{x}}_t$ are power variables or integers (depending on what they should be in the previous construction). 
    There are fewer than $t$ power variables, and so we can use our quantifier elimination to either show there is no solution $\tilde{\mathbf{x}}$ (letting us output false) or find a solution $\tilde{\mathbf{x}}$. 
    From this solution, fix $\mathbf{x}_k \coloneqq \tilde{\mathbf{x}}_k$ for $b+1 \le k \le t$.

    For a modular constraint with constant modulus, we have that for any choice of power variables $\mathbf{x}_1,\dots,\mathbf{x}_b$ of the correct form satisfying $\mathbf{x}_k \equiv \tilde{\mathbf{x}}_k \pmod{e_k}$, the constraint is satisfied by construction. 
    For the modular constraints with variable modulus, there are two cases.
    First, if $c_{i,k} = 0$ for $1 \le k \le b$, the constraint is again automatically satisfied. 
    Secondly, for the other constraints, one-by-one choose the value of $\mathbf{x}_k$ modulo the (now fixed) moduli $z_{k_i}^{d+n_{k_i}}$. 
    Our partial Gaussian elimination determines this order as it allows us to satisfy the modular constraints one at a time, and, by construction, the power variables whose modulus we choose later do not influence the earlier choices.
    By Lemma~\ref{lem:range modulo} and the earlier construction, we can indeed find $a_k,a_k',b_k,b_k'$ and substitutions of the following form.
    If $\mathbf{x}_k = z_j^{n_j}$ is a pure power variable, we make the substitution to $\mathbf{x}_k \mapsto z_j^{b_k}(z_j^{b'_k})^{n_j} = z_j^{b_k}\mathbf{x}_k$.
    If $\mathbf{x}_k = x_{i,j}z_j^{n_j}$ is a linear power variable outside the bottom family we take $\mathbf{x}_k \mapsto (a_k'x_{i,j}+a_k)z_j^{n_j} = a_k'\mathbf{x}_k + a_k\mathbf{x}_{k'}$, where $\mathbf{x}_{k'} = z_j^{n_j}$.
    Similarly, when $z_j^{n_j}$ is in the bottom family, we make both these substitutions. 
    Then, for any choice of power variables of the correct form, the modular constraints with constant modulus are automatically fulfilled.

    Finally, we verify the inequalities can be satisfied.
    For a row $A_i$ such that $A_{i,1} = \cdots = A_{i,b} = 0$, the quantifier elimination gives $A_i\tilde{\mathbf{x}} > 0$.
    Remove these rows.
    Moreover, as $A_{1:b}(\mathbf{x}_1,\dots,\mathbf{x}_b) > \mathbf{v} - A_{b+1:t}(\tilde{\mathbf{x}}_{b+1},\dots,\tilde{\mathbf{x}}_t)$ and all modular constraints are automatically satisfied, Corollary~\ref{cor:polyhedron no modular constraints + lin variables} gives that there are power variables of the correct form in this polyhedron, implying the theorem. 
\end{proof}

\end{document}